\documentclass[a4paper,11pt]{article}

\usepackage{microtype}
\usepackage{amssymb,amsmath,amsthm} \usepackage{graphicx}
\usepackage{xspace} 
\usepackage[T1]{fontenc}
\usepackage[mathlines]{lineno} 
\usepackage{fullpage}

\graphicspath{{figures/}}
\newtheorem{theorem}{Theorem} 
\newtheorem{lemma}[theorem]{Lemma}

\newcommand{\mm}{\Pi}

\newcommand{\region}[1]{\mathsf{region}(#1)}
\newcommand{\xregion}[1]{\ensuremath{\mathsf{region}_x(#1)}}

\newcommand{\lmm}{\ensuremath{\overline{\Pi}}}
\newcommand{\findcc}{\ensuremath{\textsc{FindCC}(\Pi)}}
\newcommand{\locatecc}{\ensuremath{\textsc{LocateCC}}}
\newcommand{\ds}{\ensuremath{\mathcal{D}_s}}
\newcommand{\lseg}{\ensuremath{\mathcal{L}}}
\newcommand{\rseg}{\ensuremath{\mathcal{R}}}
\newcommand{\mseg}{\ensuremath{\mathcal{M}}}
\newcommand{\winner}{\ensuremath{\mathcal{L}_w}}

\newcommand{\cc}{\ensuremath{\textsc{cc}(\Pi)}}
\newcommand{\incident}{\ensuremath{\textsc{Incident}(\Pi)}}
\newcommand{\outerbd}{\ensuremath{\textsc{OuterBD}(\Pi)}}

\newcommand{\bb}{\ensuremath{\mathcal{B}}}

\newbox\ProofSym \setbox\ProofSym=\hbox{%
	\unitlength=0.18ex%
	\begin{picture}(10,10) \put(0,0){\framebox(9,9){}}
	\put(0,3){\framebox(6,6){}}
	\end{picture}}

\title{Point Location in Incremental Planar Subdivisions%
}

\author{Eunjin Oh\thanks{Max Planck Institute for Informatics, Saarbr\"ucken, Germany, Email: \texttt{eoh@mpi-inf.mpg.de}}}

\begin{document}
	\date{}
	\maketitle

\begin{abstract}
We study the point location problem in incremental (possibly disconnected) planar subdivisions, that is, dynamic subdivisions allowing insertions of edges and vertices only. 
Specifically, we present an $O(n\log n)$-space data structure for this problem
that supports queries in $O(\log^2 n)$ time and updates in $O(\log n\log\log n)$ amortized time. This is the first result that achieves polylogarithmic query and update times simultaneously in incremental (possibly disconnected) planar subdivisions.
Its update time is significantly faster than the update time of the best known
data structure for fully-dynamic (possibly disconnected) planar subdivisions. 
\end{abstract}

\section{Introduction}
Given a planar subdivision, a point location query asks for finding the face 
of the subdivision containing a given query point.
The planar subdivisions for point location queries are induced
by planar embeddings of graphs. A planar subdivision consists of faces,
edges and vertices whose union coincides with the whole plane. 
An edge of a subdivision is considered to be
open, that is, it does not include its endpoints (vertices). A face of
a subdivision is a maximal connected subset of the plane that does
not contain any point on an edge or a vertex.
The boundary of a face of a subdivision may consist of several connected
components. 
Imagine that we give a direction to each edge on the boundary of a face $F$ 
so that $F$ lies to the left of it. (If an edge is incident to $F$ only, we consider it as two edges with opposite directions.)  
We call a boundary component of $F$ the
\emph{outer boundary} of $F$ if it is traversed
in counterclockwise order around $F$. 
Every bounded face has exactly one outer boundary. We call a 
connected component other than the outer boundary an \emph{inner
boundary} of $F$.

We say a planar subdivision is \emph{dynamic} if the subdivision changes
dynamically by insertions and deletions of edges and vertices.
A dynamic planar subdivision is \emph{connected} if 
the underlying graph is connected at any time. 
In other words, the boundary of each face is connected at any time.
We say a dynamic planar subdivision is \emph{general} if it is not necessarily
connected. 
There are three versions of dynamic planar subdivisions with respect to
the update operations they support: incremental,
decremental and fully-dynamic. An incremental subdivision
allows only insertions of edges and vertices, and a decremental
subdivision allows only deletions of edges and vertices. 
A fully-dynamic subdivision allows both of them.

The dynamic point location problem is closely related to the dynamic 
vertical ray shooting problem in the case of connected subdivisions~\cite{Cheng-NewResults-1992}.  In this 
problem, we are asked to find the edge of a dynamic 
planar subdivision that lies immediately above a query point. 
The boundary of each face in a dynamic connected subdivision is connected,
so one can maintain the boundary of each face efficiently using a concatenable queue.
Then one can
answer a point location query without increasing the space and time
complexities using a data structure for the dynamic vertical ray shooting
problem~\cite{Cheng-NewResults-1992}.

However, it is not the case in general  planar
subdivisions.  Although the dynamic vertical ray shooting data structures 
presented
in~\cite{ABG-Improved-2006,BJM-Dynamic-1994,cn-locatoin-2015,Cheng-NewResults-1992}
work for general  subdivisions, it is unclear how one can
use them to support point location queries efficiently.  As pointed
out in some previous
works~\cite{cn-locatoin-2015,Cheng-NewResults-1992}, a main issue
concerns how to test for any two edges if they belong to the boundary of the same face in the subdivision. 
This is because the boundary of a face may
consist of more than one connected component.


\paragraph{Previous work.}
There are several data structures for the point location problem
in \emph{fully-dynamic} planar \emph{connected} subdivisions~\cite{ABG-Improved-2006,BJM-Dynamic-1994,cn-locatoin-2015,Cheng-NewResults-1992, CPT-unified-1996,CT-monotone-1992,GT-dynamictree-1998,PT-monotone-1989}.
None of the known results for
this problem is superior to the others, and optimal update and
query times are not known. 
The latest result was given by Chan and
Nekrich~\cite{cn-locatoin-2015}. The linear-size data structure by Chan and
Nekrich~\cite{cn-locatoin-2015} supports $O(\log n(\log \log n)^2)$
query time and $O(\log n \log \log n)$ update time in the pointer
machine model, where $n$ is the number of the edges of the current subdivision.
Some of them~\cite{ABG-Improved-2006,BJM-Dynamic-1994,cn-locatoin-2015,Cheng-NewResults-1992} including the result by Chan and Neckrich can be used for 
answering vertical ray shooting queries without increasing the running
time. 

There are data structures 
for answering point location queries more efficiently
in \emph{incremental} planar \emph{connected} subdivisions in the pointer machine model~\cite{ABG-Improved-2006, GT-dynamictree-1998, IMAI19871}. The best known data structure supports
$O(\log n\log^* n)$ query time and $O(\log n)$ amortized update time~\cite{ABG-Improved-2006} and has linear size. 
This data structure can be modified
to support $O(\log n)$ query time and $O(\log^{1+\epsilon} n)$ amortized update time
for any $\epsilon>0$.
In the case that every cell is monotone at any time,
there is a linear-size data structure supporting $O(\log n\log\log n)$ query time and $O(1)$ amortized update time~\cite{GT-dynamictree-1998}.

On the other hand,
little has been known about this problem in \emph{fully-dynamic}
planar \emph{general} subdivisions, which was recently mentioned by Snoeyink~\cite{handbook}. 
Very recently, Oh and Ahn~\cite{Oh-2018} presented a linear-size data structure
for answering point location queries in $O(\log n(\log\log n)^2)$ time with $O(\sqrt{n}\log n(\log\log n)^{3/2})$ amortized update time. In fact, this is the only data structure
known for answering point location queries in general dynamic planar subdivisions.
In the same paper, the authors also considered the point location problem
in decremental general subdivisions. They presented a linear-size data structure
supporting $O(\log n)$ query time and $O(\alpha(n))$ update time, where $n$
is the number of the edges in the initial subdivision and $\alpha(n)$ is
the inverse Ackermann function. 

\paragraph{Our result.}
In this paper, we present a data structure for answering
point location queries in incremental  general planar subdivisions in the pointer machine model.
The data structure supports $O(\log^2 n)$ query time
and $O(\log n\log\log n)$ amortized update time, where $n$ is the number
of the edges at the current subdivision. The size of the data structure is 
$O(n\log n)$. 
This is the first result on the point location problem specialized in incremental general 
planar  subdivisions. The update time of this data structure is significantly faster than the update time of the data structure
in fully-dynamic planar general subdivisions in~\cite{Oh-2018}.

\paragraph{Comparison to the decremental case.} 
In decremental general subdivisions, there is a simple and efficient data structure
for point location queries~\cite{Oh-2018}. This data structure maintains
the decremental subdivision explicitly: for each face $F$ of the subdivision, it 
maintains a number of  concatenable queues each of which
stores the edges of each connected component of the boundary of $F$. 
When an edge is removed, two faces might be merged into one face, but
no face is subdivided into two faces. Using this property, they 
 maintain a disjoint-set data structure for each face such that
an element of the disjoint-set data structure is the name of a concatenable queue
representing a connected component of the boundary of this face.

In contrast to decremental subdivisions, it is unclear how to maintain an incremental
subdivision explicitly. Suppose that a face 
$F$ is subdivided into two faces $F_1$ and $F_2$ by the insertion of an
edge $e$. See Figure~\ref{fig:inc-dec}(a). An inner boundary of $F$ becomes
an inner boundary of either $F_1$ or $F_2$ after $e$ is inserted.
It is unclear how to update the set of the inner boundaries of $F_i$ for $i=1,2$  
without accessing every queue representing an inner boundary of $F$.
If we access all such concatenable queues, the total insertion time for $n$
insert operations is 
$\Omega(n^2)$ in the worst case. 
Therefore it does not seem that the approach in~\cite{Oh-2018} works for incremental
subdivisions.

\begin{figure}
  \begin{center}
    \includegraphics[width=0.7\textwidth]{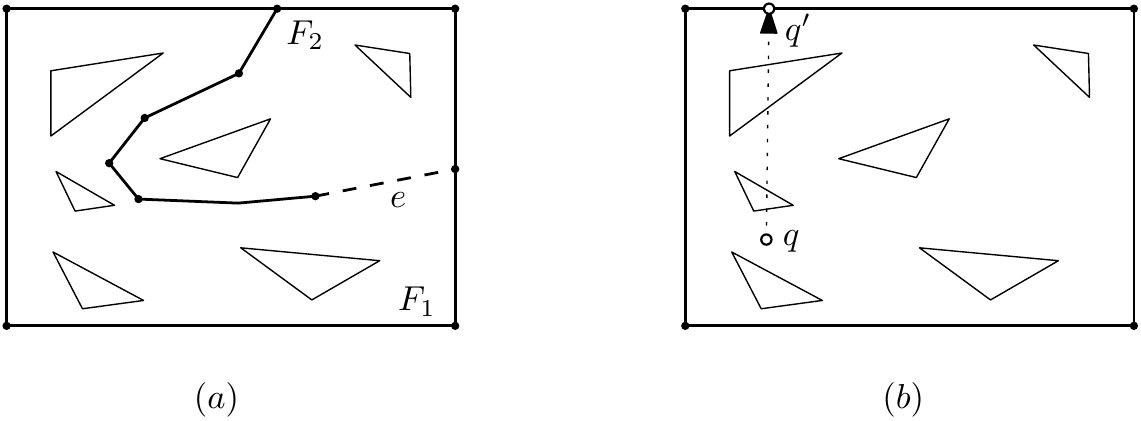}
    \caption{\small (a) The insertion of $e$ makes the face subdivided into
    	two subfaces $F_1$ and $F_2$. Then $F_1$ has four inner boundaries,
    	and $F_2$ has two inner boundaries. (b) Given a query point $q$,
    	imagine that we shoot the upward vertical ray from $q$
    	which penetrates inner boundaries not containing $q$ 
    	until it hits the outer
    	boundary of some face. Once we obtain the point $q'$ where the ray reaches, we can find the face
    	containing $q$ efficiently.
    	\label{fig:inc-dec}}
  \end{center}
\end{figure}

\paragraph{Outline.}
Instead of maintaining the whole subdivision
(i.e., all boundary components for each face) explicitly, we maintain the outer boundary of a face only using a concatenable queue.
We define the name of each face to be the name of the concatenable queue 
representing the outer boundary of the face. 
Thus if we have an outer boundary edge of the face containing a query point $q$,
we can return the name of the face immediately.
Note that, in connected subdivisions, the edge lying immediately above $q$   is such an edge. However it is not the case in general subdivisions.
In our query algorithm, we shoot a vertical upward ray from $q$ 
which penetrates boundary components not containing $q$ in their interiors 
until it hits the outer boundary of a face $F$.  
Observe that $F$ contains $q$. See Figure~\ref{fig:inc-dec}(b). Then we can return the name of $F$. Specifically, our two-step query algorithm works as follows.

First, we find the connected component $\gamma$ of the underlying graph of the current subdivision $\Pi$ 
which contains the outer boundary of the face containing the query point.
To do this, we reduce the problem into a variant of the stabbing query problem, which
we call the \emph{stabbing-lowest query problem} for trapezoids.
Consider the vertical decomposition
of the subdivision induced by each connected component of the underlying graph of $\Pi$. There are $O(n)$ cells of the 
vertical decompositions for every connected component of $\Pi$ in total, where
$n$ is the number of the edges in $\Pi$.
Let $\Box$ be the cell (trapezoid) whose upper side lies immediately above the query point among the cells containing the query point.
Then the connected component from which $\Box$ comes contains the outer boundary
of the face containing the query point.
However, it takes $\Omega(n)$ time to update the vertical subdivisions
for each edge insertion in the worst case. We present an alternative way to obtain trapezoids satisfying this property and allowing efficient update time.

Second, we find the face of the subdivision $\Pi_\gamma$ induced by $\gamma$.
Note that  $\Pi_\gamma$ is connected. 
Also, the boundary of the face of $\Pi_\gamma$ containing the query point 
coincides with the outer boundary of the face $F$ of $\Pi$ containing the query point.
Therefore, we can find the name of $F$ by applying a point location query on $\Pi_\gamma$.
To do this, we maintain a point location data structure on $\Pi_{\gamma'}$ for each connected component $\gamma'$. Note that two connected components might be merged
into one. To resolve this issue, we present a new data structure supporting an efficient  merge operation, which is a variant of the dynamic data structure by Arge et al.~\cite{ABG-Improved-2006}.

\section{Preliminaries}\label{sec:preliminary}
Consider an incremental planar subdivision $\Pi$.
We use $\lmm$ to denote the union of the edges and vertices of $\mm$.
We require that
every edge of $\Pi$ be a straight line segment.  
For a set $A$ of elements (points or edges), we use $|A|$ to denote the number of the
elements in $A$. For a planar subdivision $\Pi'$, we use $|\Pi'|$ to denote the 
complexity of $\Pi'$, i.e., the number of the edges of $\Pi'$.
We use $n$ to denote the number of the edges of $\mm$ at the moment. 
For a connected component $\gamma$ of $\lmm$, we use $\Pi_\gamma$ to denote
the subdivision induced by $\gamma$. Notice that it is connected.

In this problem, we are to process a mixed sequence of $n$ edge insertions and
vertex insertions
so that given a query point $q$ the face of the current subdivision 
containing $q$ can be computed efficiently.
More specifically, each face  in the subdivision is assigned a distinct
name, and given a query point
the name of the face containing the point is to be reported.
For the insertion of an edge $e$, we require $e$ to intersect no edge or
vertex in the current subdivision. Also, an endpoint of $e$ is required to lie on a
face or a vertex of the subdivision. We insert the endpoints of $e$ in the subdivision
as vertices if they were not vertices of the subdivision. 
For the insertion of a vertex $v$, it is required to lie on an edge or a face of the current subdivision. If it lies on an edge, the edge is split into two
(sub)edges whose common endpoint is $v$. 

\subsection{Tools}
In this subsection, we introduce several tools we use
in this paper. 
A \emph{concatenable queue} represents a sequence of elements, and allows five  operations: insert an 
element, delete an element, search an element, split a sequence into two subsequences,
and concatenate two sequences into
one. By implementing it with a 2-3 tree, we can support
each operation in $O(\log N)$ time, where $N$ is the number of
elements at the moment. 

The \emph{vertical decomposition} of a (static) planar subdivision $\Pi_s$
is a finer subdivision of $\Pi_s$ by adding a number of vertical line segments.
For each vertex $v$ of $\Pi_s$, consider two vertical extensions
from $v$, one going upwards and one
going downwards. The extensions stop when they meet an edge
of $\Pi_s$ other than the edges incident to $v$. 
The vertical decomposition of $\Pi_s$ is the subdivision
induced by the vertical extensions contained in 
the \emph{bounded} faces of $\Pi_s$ together with the edges of $\Pi_s$. Note that the unbounded face of $\Pi_s$ remains the same in the vertical decomposition.
In this paper, we do not consider the unbounded face of $\Pi_s$
as a cell of the vertical decomposition. Therefore, every cell is
a trapezoid or a triangle (a degenerate trapezoid). 
There are $O(|\Pi_s|)$ trapezoids in the vertical decomposition of $\Pi_s$. We treat each trapezoid as a closed set.
We can compute the vertical decomposition 
in $O(|\Pi_s|)$ time~\cite{triangulation} since we decompose the bounded faces only. 

We use segment trees, interval trees and priority search trees as basic building
blocks of our data structures. In the following, we briefly review those trees.
For more information, refer to~\cite[Section 10]{CGbook}.

\paragraph{Segment and interval trees.}
We first introduce 
the segment and the interval trees on a set $\mathcal{I}$ 
of $n$ intervals on the $x$-axis.
Let $\mathcal{I}_p$ be the set of the endpoints of the intervals of $\mathcal{I}$. 
The base structure of the segment and interval trees is a binary search tree on $\mathcal{I}_p$
of height $O(\log n)$ such that each leaf node corresponds to exactly one point of $\mathcal{I}_p$.
Each internal node $v$ corresponds to a point $\ell(v)$ on the $x$-axis  
and an interval $\region{v}$ on the $x$-axis such that $\ell(v)$ is 
the midpoint of $\mathcal{I}_p\cap \region{v}$. For the root $v$, $\region{v}$ is defined as the $x$-axis. 
Suppose that $\ell(v)$ and $\region{v}$ are 
defined for a node $v$. For its two children $v_\ell$ and $v_r$, $\region{v_\ell}$ and $\region{v_r}$ are the left and right
regions, respectively, in the subdivision of $\region{v}$ 
induced by $\ell(v)$.

For the interval tree, each interval $I\in\mathcal{I}$ is stored in exactly one node: the node $v$ of maximum depth such that $\region{v}$ contains $I$. 
In other words, it is stored in the lowest common ancestor of two leaf nodes
corresponding to the endpoints of $I$.
For the segment tree, each interval $I$ is stored in $O(\log n)$ nodes:
the nodes $v$ such that $\region{v}\subseteq I$, but $\region{u}\not\subseteq I$ 
for the parent $u$ of $v$. 
For any point $p$ on the $x$-axis, let $\pi(p)$ be the 
search path of $p$ in the base tree.
Each interval of $\mathcal{I}$ containing $p$ is stored in some nodes of $\pi(p)$  in both trees. However, not every interval stored in the nodes of $\pi(p)$ contains $p$ in the interval tree 
while every interval stored in the nodes of $\pi(p)$ contains $p$ in the segment tree.

Similarly, the segment tree and the interval tree on a set $\mathcal{S}$
of $n$ line segments in the plane are defined as follows.
Let $\mathcal{S}_x$ be the set of the projections of the line segments
of $\mathcal{S}$ onto the $x$-axis. 
The segment and interval trees of $\mathcal{S}$ are 
basically the segment and interval trees on $\mathcal{S}_x$, respectively.
The only difference is that instead of storing the projections, 
we store a line segment of $\mathcal{S}$ in the nodes where
its projection is stored in the case of $\mathcal{S}_x$.
As a result, $\ell_x(v)$ and $\xregion{v}$ for the trees of $\mathcal{S}$
 are naturally defined as the vertical line containing
 $\ell(v)$ and the smallest vertical slab containing $\region{v}$ for the
 trees of $\mathcal{S}_x$, respectively.
 If it is clear in context, we use $\ell(v)$ and $\region{v}$
 to denote $\ell_x(v)$ and $\xregion{v}$, respectively.

\paragraph{Interval tree with larger fan-out.}
To speed up updates and queries, we use an interval tree with
larger fan-out $f\geq 2$ for storing the intervals of $\mathcal{I}$. 
As the binary case mentioned above, it is naturally extended to the one
for line segments in the plane. 
The base tree is a balanced search tree of $\mathcal{I}_p$ with fan-out $f$,
which has height of $O(\log n/\log f)$.
Then each node $v$ of the base tree has at most $f$ children $u_1,\ldots,u_{f'}$ 
 and has an interval $\region{v}$ satisfying that 
 the left and right endpoints of $\region{u_i}$ are the $(i-1)$th
 and $i$th $f'$-quantile of
 $\mathcal{I}_p\cap \region{v}$, respectively, 
  for $1\leq i\leq f'$.

Each node $v$ of the base tree has three sets $\lseg(v)$, $\rseg(v)$ and $\mseg(v)$
of intervals of $\mathcal{I}$. 
An interval $I\in\mathcal{I}$ is stored in at most three nodes as follows. 
Let $v$ be the node of maximum depth such that $\region{v}$ contains $I$. 
Let $u_1$ and $u_2$ be the children of $v$ such that $\region{u_1}$ contains 
the left endpoint of $I$ and $\region{u_2}$ contains the right endpoint of $I$.
We store $I$ in $\lseg(u_1)$, $\rseg(u_2)$ and $\mseg(v)$.
Precisely, we store $I\cap \region{u_1}$ in $\lseg(u_1)$,  $I\cap \region{u_2}$
in $\rseg(u_2)$, and the remaining piece of $I$ in $\mseg(v)$.
Then every piece stored in $\lseg(v)$ (and $\rseg(v)$) has a common endpoint.
For the pieces stored in $\mseg(v)$, their endpoints have at most $f$ distinct
$x$-coordinates. We will make use of these properties to speed up updates and queries
in Section~\ref{sec:D2} and Section~\ref{sec:stabbing}.
In the following, to make the description easier, we do not distinguish
a piece stored in a set and the line segment of $\mathcal{I}$ from which the piece
comes.

\paragraph{Priority search tree.}
Suppose that we are given a set $\mathcal{S}_\ell$ of $n$ line segments in the plane
having their left endpoints on a common vertical line $\ell$.
Such edges can be sorted in $y$-order: from top to bottom with respect
to their endpoints on $\ell$.
The \emph{priority search tree} can be used to answer 
vertical ray shooting queries efficiently in this case. 
The base tree is a binary search tree of height $O(\log n)$ on
the endpoints of the line segments 
of $\mathcal{S}_\ell$ on $\ell$.
Each line segment corresponds to a leaf node of the base tree. 
Each node $v$ stores the $x$-coordinate of 
the right endpoint of the line segment with rightmost right endpoint
as its key among all line segments
corresponding to the leaf nodes of the subtree rooted at $v$.
Cheng and Janardan~\cite{Cheng-NewResults-1992} showed that
a vertical ray shooting query can be answered in time linear in the height of the base tree
by traversing two paths from the root to leaf nodes.
	 
In our problem, 
an advantage for using the priority search tree is that it can be constructed in
linear time if the line segments of $\mathcal{S}_\ell$ are sorted
with respect to their $y$-order. To see this, observe that the base tree
can be constructed in linear time in this case. Then we compute the key
for each node of the base tree in a bottom-up fashion.  
Using this property, we can merge two priority search trees efficiently.

\subsection{Subproblem: Stabbing-Lowest Query Problem for Trapezoids} 
The trapezoids we consider in this paper 
have two sides parallel to the $y$-axis unless otherwise stated.
We consider the incremental \emph{stabbing-lowest query problem} for trapezoids
as a subproblem. In this problem, we are given
a set $\mathcal{T}$ of trapezoids which is initially empty and changes dynamically
by insertions of trapezoids.
Here, the trapezoids we are given satisfy that no two upper or lower sides
of the trapezoids cross each other. But it is possible that the upper (or lower) side of one trapezoid crosses a vertical side of another trapezoid.
 We process a sequence of updates 
for the following task.  Given a query point $q$, 
the task is to find the trapezoid whose upper side lies immediately above $q$
among all trapezoids of $\mathcal{T}$ containing $q$. We call such a trapezoid the
\emph{lowest trapezoid stabbed by $q$}.

%
In Section~\ref{sec:stabbing}, we present a data structure
for this problem 
in the case that only insertions are allowed.
The worst case query time is $O(\log^2 n)$, the amortized update time is $O(\log n\log\log n)$,
and the size of the data structures is $O(n\log n)$. 
We will use this data structure as a black box in Section~\ref{sec:main}.

\section{Point Location in Incremental General Planar Subdivisions}\label{sec:main}
Compared to connected subdivisions,
a main difficulty for handling dynamic general planar subdivisions lies
in finding the faces incident to the edge $e$ lying immediately
above a query point~\cite{Cheng-NewResults-1992}.  
If $e$ is contained in the outer boundary of a face, we can find
the face as the algorithm in~\cite{Cheng-NewResults-1992} for connected planar subdivisions does.
However, this approach does not work if $e$ lies on an inner boundary of a face. 
To overcome this difficulty, instead of finding the edge in $\mm$ lying immediately above a query point $q$, we find an outer boundary edge of 
the face $F$ of $\mm$ containing $q$.
To do this, we answer a point location query in two steps.

First, we find the (maximal) connected component $\gamma$ 
of $\lmm$ containing an outer boundary edge of $F$.
We use $\findcc$ to denote this data structure. 
 We observe that
the outer boundary of the face of $\Pi_\gamma$ containing $q$  
coincides with the outer boundary of $F$. 
We maintain the outer boundary of each face in a concatenable queue.
Thus given an outer boundary edge of $F$, we can return
the name of $F$ by defining  the name of each face of $\mm$ as the name of the concatenable queue representing its  
outer boundary.

Second, we apply a point location query on $\Pi_\gamma$.
More specifically, we find the face $F_\gamma$ in $\Pi_\gamma$ containing $q$, find the concatenable queue representing the outer boundary of $F_\gamma$,
and return its name.
Since $\Pi_\gamma$ is connected, we can maintain an efficient data structure 
for point location queries on $\Pi_\gamma$.
We use $\locatecc(\gamma)$ to denote this data structure. 
Each of Sections~\ref{sec:D1} and~\ref{sec:D2}
describes each of 
the two data structures together with query and update algorithms. 

In addition to them, we maintain the following data structures:
$\incident$ for
checking if a new edge is incident to $\lmm$,
and $\cc$ for maintaining the connected components of $\lmm$,
and $\outerbd$ for maintaining the outer boundary of each face of $\mm$.
The update times for these structures are subsumed by the total update time.

\paragraph{$\incident$: For checking if a new edge is incident to $\lmm$.}
To check if a new edge $e$ is incident to a connected component of $\lmm$, we 
maintain a balanced binary search tree on the vertices of $\mm$ 
in the lexicographical order with respect to their $x$-coordinates and then their $y$-coordinates.  Also, for each vertex of $\mm$, we maintain
a balanced binary search tree on the edges incident to it in clockwise order
around the vertex. When an edge or a vertex is inserted, we can update these data structures
in $O(\log n)$ time. Since each endpoint of $e$ lies on a vertex
of $\mm$ or in a face of $\mm$, we can check if $e$ is incident to a connected component of $\lmm$ in $O(\log n)$ time. 

An edge $e$ is stored in two balanced binary search trees: each for its endpoint.
We make the elements in the trees corresponding to $e$ point to each other.
Also, we make an element in each balanced binary search tree point to its successor
and predecessor. In this way, we can traverse the edges of the outer boundary of a face of $\Pi$ from a given edge in clockwise order in time linear in the number of the edges.

\paragraph{$\cc$: For maintaining the connected components of $\lmm$.}
We maintain each connected component of $\lmm$ using
a disjoint-set data structure~\cite{Tarjan-1975}.
A disjoint-set data structure keeps track of
a set of elements partitioned into a number of
disjoint subsets. It has size linear in the total number of
elements, and can be used to check if two elements
are in the same partition and to merge two partitions into one.
Both operations can be done in $O(\alpha(N))$ time, where $N$ is the
number of elements at the moment and $\alpha(\cdot)$
is the inverse Ackermann function.
In our case, we store the edges of $\mm$ 
to a disjoint-set data structure, and we say that
two edges are in the same partition
if and only if they are in the same connected component of $\lmm$.
In this way, we can check if two edges are in the same connected component 
of $\lmm$ in $O(\alpha(n))$ time. The update time for each
insertion is $O(\log n)$ since we need to find the connected components incident
to the new edge using $\incident$.

\paragraph{$\outerbd$: For maintaining the outer boundary of each face of $\mm$.}
We maintain concatenable queues each of which represents
the outer boundary of a face of $\mm$. Also, we maintain a set $\mathcal{E}$ of the edges of $\mm$, and let an edge $e$ of $\mm$ point to its corresponding element in (at most two) concatenable queues 
so that we can return the name of each concatenable queue which $e$ belongs to in constant time once we have the pointer pointing to the element  in $\mathcal{E}$ corresponding to $e$.

There are only two cases that the outer boundary of a face changes
by the insertion of a new edge $e$: (1) both endpoints of $e$ are contained in the same connected component of $\lmm$, or (2) they are contained in distinct connected components of $\lmm$. See Figure~\ref{fig:concatenable}.  Using $\cc$ and $\incident$, we can check if the insertion of an edge $e$ belongs to each of the cases in $O(\log n)$ time.
 Let $F$ be the face containing $e$. 
 
\begin{figure}
	\begin{center}
		\includegraphics[width=0.7\textwidth]{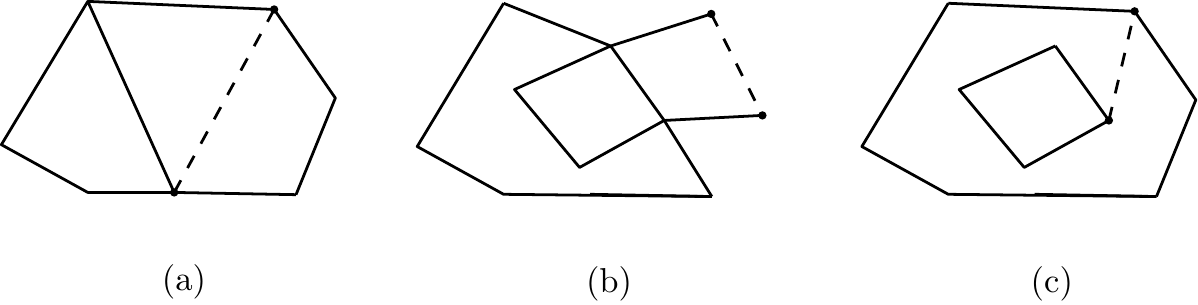}
		\caption{\small The dashed line segment is $e$. (a) Case~(1). 
			We split a concatenable queue into two concatenable queues.
			(b) Case~(1). We create a new concatenable queue.
			(c) Case~(2). We insert the five edges including $e$ into
			a concatenable queue.
			\label{fig:concatenable}}
	\end{center}
\end{figure}

Consider Case~(1). 
 We check if the endpoints of $e$ lie on the outer boundary of $F$ by finding the edges incident to each endpoint of $e$ that comes before and after $e$ around $v$ using $\cc$ and $\incident$.
If so, the face $F$ is subdivided into two faces. See Figure~\ref{fig:concatenable}(a). We split the concatenable queue for $F$ into two queues in $O(\log n)$ time. Otherwise, a new face containing $e$ on its outer boundary appears. See Figure~\ref{fig:concatenable}(b). Then we trace the inner boundary of $F$ incident to $e$ 
in time linear in its size using $\incident$, and make a new concatenable queue for this face. This takes $(N\log n)$ time, where $N$ is the size of the outer boundary of the new face.

Consider Case~(2). In this case, using $\cc$ and $\incident$, we check if one of the endpoints of $e$ is contained in the outer boundary of $F$, and the other is contained in an inner boundary of $F$. This is the only case that a new face appears. See Figure~\ref{fig:concatenable}(c). If so, the new face, which is $F\setminus e$, has the outer boundary which is the union of the outer boundary of $F$, the inner boundary of $F$ incident to $e$, and $e$. Then we trace such an inner boundary of $F$ in time linear in its size, and insert them the concatenable queue for $F$ one by one, and then insert $e$. This takes $(N\log n)$ time, where $N$ is the size of the inner boundary of $F$ incident to $e$.

The total time for maintaining the concatenable queues is $O(n\log n)$.
This is because each edge $e$ of $\mm$ is inserted to some concatenable queues
at most twice. Consider any two faces $F_1$ and $F_2$ containing $e$ on their outer boundaries and lying locally below $e$ which appear in the course of updates. Assume that $F_1$ appears before $F_2$ appears. This means that $e$ has become an outer boundary edge of $F_2$ by a series of splits of faces from $F_1$. In the course of these splits, the concatenable queues change only by the split operation, which takes $O(\log n)$ time per edge insertion. Therefore, the amortized time for maintaining the concatenable queues is $O(\log n)$.

\subsection{$\findcc$: Finding One Connected Component for a Query Point}\label{sec:D1}
We construct a data structure for finding the (maximal) connected component $\gamma_q$
of $\lmm$ containing the outer boundary of the face of $\Pi$ containing a query point $q$. 
To do this, we compute a set $\mathcal{T}$ of $O(n)$ trapezoids each of which \emph{belongs} to exactly one edge of $\mm$ such that
the edge to which the lowest trapezoid stabbed by $q$ belongs is contained in $\gamma_q$.
Then we construct the stabbing-lowest data structure on $\mathcal{T}$ described in
Section~\ref{sec:stabbing}.

\begin{figure}
	\begin{center}
		\includegraphics[width=0.8\textwidth]{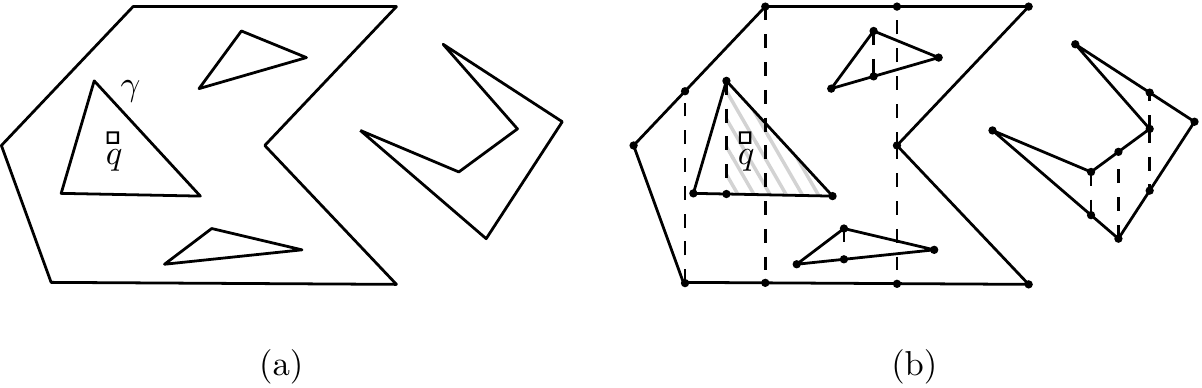}
		\caption{\small (a) The component $\gamma$ contains
			the outer boundary of the face containing $q$.
			(b) Using the vertical decomposition, we obtain $O(n)$ (possibly intersecting) trapezoids.
			Their corners are marked with disks.
			The lowest trapezoid stabbed by $q$ is the dashed one, which comes from $\gamma$.
			\label{fig:vertical-decompose}}
	\end{center}
\end{figure}

\subsubsection{Data Structure and Query Algorithm}
For each connected component $\gamma$ of $\lmm$,
consider the subdivision $\Pi_\gamma$ induced by $\gamma$.
Notice that $\Pi_\gamma$ is connected. Let $U(\gamma)$ be the union of the closures 
of all bounded faces of $\Pi_\gamma$.  Note that it might be 
disconnected and contain an edge of $\gamma$ in its interior.  
Imagine that we have the cells (trapezoids) of the vertical decomposition of $U(\gamma)$. 
We say that a cell \emph{belongs} to the edge of $\gamma$ containing the
upper side of the cell. 
Let $\mathcal{T}_\gamma$ be the set of the cells (trapezoids) for $\gamma$,
and $\mathcal{T}$ be the union of $\mathcal{T}_\gamma$ 
for every connected component $\gamma$ of $\lmm$.
See Figure~\ref{fig:vertical-decompose}. 
We will show in Lemma~\ref{lem:subproblem} that a generalized version of the following statement holds: 
the lowest trapezoid in $\mathcal{T}$ stabbed by a query point $q$
belongs to an edge of $\gamma_q$. If no trapezoid in $\mathcal{T}$ contains $q$, 
the query point is contained in the unbounded face of $\Pi$.

However, each edge insertion may induce 
$\Omega(n)$ changes on $\mathcal{T}$ in the worst case.
For an efficient update procedure, we define and
construct the trapezoid set $\mathcal{T}_\gamma$ in a slightly different way by allowing some edges lying inside $U(\gamma)$ to define
trapezoids in $\mathcal{T}_\gamma$. 
	For a connected component $\gamma$ of $\lmm$, we say a set of connected subdivisions induced by edges of $\gamma$ 
	\emph{covers} $\gamma$ if an edge of $\gamma$ is contained in at most two subdivisions, and
	one of the subdivisions contains all edges of the boundary of $U(\gamma)$.
Let $\mathcal{F}_\gamma$ be a set of connected subdivisions covering $\gamma$.   See Figure~\ref{fig:vertical-decompose-finer}. 
Notice that $\mathcal{F}_\gamma$ is not necessarily unique. 
For a technical reason, if the union of some edges (including their endpoints) in a subdivision of $\mathcal{F}_\gamma$  forms
a line segment, we treat them as one edge.
Then we let $\mathcal{T}_\gamma$ be the set of the cells of the 
vertical decompositions of the subdivisions in $\mathcal{F}_\gamma$. 
We say that a cell (trapezoid) of $\mathcal{T}_\gamma$
\emph{belongs} to the edge of $\gamma$ containing the upper side of the cell.
Let $\mathcal{T}$ be the union of all such sets $\mathcal{T}_\gamma$.

The following lemma shows that 
the lowest trapezoid in $\mathcal{T}$ stabbed by $q$
belongs to an edge of $\gamma_q$.
Thus by constructing a stabbing-lowest data structure on $\mathcal{T}$, 
we can find $\gamma_q$ in $O(Q(n))$ time, 
where $Q(n)$ is the 
query time for answering a stabbing-lowest query.
The query time of the stabbing-lowest data structure on $n$ trapezoids 
described in Section~\ref{sec:stabbing} is $O(\log^2 n)$.

\begin{lemma}\label{lem:subproblem}
	The lowest trapezoid in $\mathcal{T}$ stabbed by a query point $q$
	belongs to an edge of the connected component of $\lmm$ containing the outer boundary
	of the face of $\mm$ containing $q$. If the face of $\mm$ containing $q$
	is unbounded, no trapezoid in $\mathcal{T}$ contains $q$.
\end{lemma}
\begin{proof}
	We first claim that a trapezoid $\Box_q$ in $\mathcal{T}$ 
	belonging to an edge in $\gamma_q$ contains $q$ if $F_q$ is a bounded face,
	 where $\gamma_q$ is the 
	connected component of $\lmm$ containing the outer boundary
	of the face $F_q$ containing $q$.
	By definition, there is a connected subdivision $\Pi'$ in $\mathcal{F}_{\gamma_q}$ containing 
	all edges of the boundary of $U(\gamma_q)$.
	Thus $q$ is contained in a bounded face of $\Pi'$ in its closure.
	Since the cells of the vertical decomposition of $\Pi'$ are contained
	in $\mathcal{T}$, one of them contains $q$.
	
	Then we claim that any trapezoid $\Box$ containing $q$ and belonging
	to an edge on $\lmm \setminus \gamma_q$ 
	has
	the upper side lying above the upper side of $\Box_q$ (i.e., 
	the vertical upward ray from $q$ intersects the upper side of $\Box_q$ before
	intersecting the upper side of $\Box$.) This claim implies the lemma in the case that $F_q$ is bounded.
	Assume to the contrary that the upper side of $\Box$ lies below the upper side of $\Box_q$. 
	Let $\gamma$ be the connected component of $\lmm$ containing the edge to which $\Box$ belongs. 
	Since $\Box$ contains $q$, the subdivision induced by $\gamma$ has
	a bounded face containing $q$ in its closure.
	This means that the outer boundary of $F_q$ is contained in 
	the closed region bounded by the outer boundary of this bounded face.
	Notice that, $\gamma$ and $\gamma_q$ are disjoint since they are maximal connected components of $\lmm$. Moreover, $\gamma$ is contained
	in the interior of $U(\gamma_q)$, which contradicts
	that $\gamma_q$ contains the outer boundary of $F_q$.
	
	Now consider the case that $F_q$ is the unique unbounded face of $\mm$. 
	For any connected subdivision induced by edges of $\mm$, the unique unbounded face contains $q$. Therefore, no trapezoid of $\mathcal{T}$
	contains $q$. This proves the lemma. 
\end{proof}

\begin{figure}
	\begin{center}
		\includegraphics[width=0.8\textwidth]{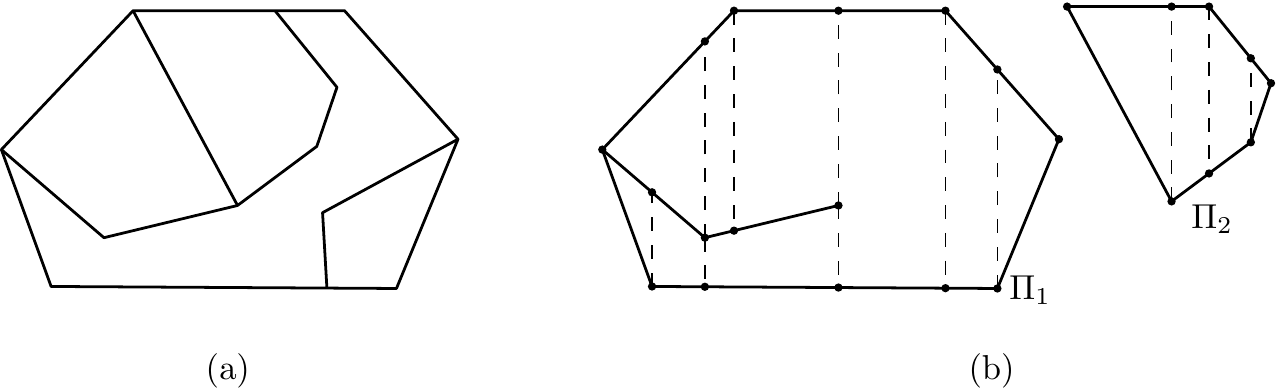}
		\caption{\small (a) A connected component $\gamma$. (b) A set of two subdivisions covering $\gamma$. The set $\mathcal{T}_\gamma$ consists of  
			the trapezoids
			in the vertical decompositions of $\Pi_1$ and $\Pi_2$.
			\label{fig:vertical-decompose-finer}}
	\end{center}
\end{figure}

The following lemma shows that the size of $\findcc$  is $O(S(n))$,
where $S(n)$ is the size of a stabbing-lowest data structure.

\begin{lemma}
	The size of $\mathcal{T}$ is $O(n)$, where $n$ is the complexity of the current subdivision.\label{lem:size-D1}
\end{lemma}
\begin{proof}
	We first claim that the total complexity of the
	subdivisions of $\mathcal{F}_\gamma$ for every connected component
	$\gamma$ of $\lmm$ is $O(n)$. This is simply because
	for each connected component $\gamma$ of $\lmm$, each edge of $\gamma$
	is contained in at most two subdivisions of $\mathcal{F}_\gamma$ by definition.
	Recall that $\mathcal{T}$ is the cells of the vertical decomposition
	of the subdivisions of $\mathcal{F}_\gamma$. 
	The vertical decomposition of a planar subdivision has $O(N)$ cells,
	where $N$ is the complexity of the planar subdivision.
	Therefore, the size of $\mathcal{T}$ is $O(n)$.
\end{proof}

\begin{lemma}
	Given the data structure $\findcc$ of size $O(n)$,
	we can find the connected component of $\lmm$ containing
	the outer boundary of the face of $\mm$ containing a query point 
	in $O(\log^2 n)$ time.
\end{lemma}
\subsubsection{Update Algorithm}
We maintain a stabbing-lowest data structure on $\mathcal{T}$.
As we did in the previous section, let $\mathcal{T}_\gamma$ be the 
set of the cells of $\mathcal{T}$ which belongs to an edge of $\gamma$
for each connected component $\gamma$ of $\mm$. 
Let $\mathcal{F}_\gamma$ be a set of the connected subdivisions 
such that $\mathcal{T}_\gamma$ consists of the cells of the 
vertical decomposition of the subdivisions of the set. 
Once $\mathcal{F}_\gamma$ covers 
$\mm_\gamma$ for every connected component of $\gamma$
of $\mm$, the query algorithm takes $O(\log^2 n)$ time. 
In this section, we show how to update $\mathcal{T}$ so that $\mathcal{F}_\gamma$ covers
$\mm_\gamma$ for every connected component. 
But we do not maintain
the sets $\mathcal{F}_\gamma$ and $\mathcal{T}_\gamma$ 
for a connected component $\gamma$ of $\lmm$.
We use them only for description purpose.

For the insertion of a vertex $v$, we do not need to do anything 
for $\findcc$. To see this,
observe that $\lmm$ remains the same after the insertion of $v$. The insertion of $v$
splits one edge, say $e$, into two edges, say $e_1$ and $e_2$. 
For the connected component $\gamma$ of $\lmm$ containing $e$, we have a set $\mathcal{F}_\gamma$ of subdivisions covering $\gamma$. The edge $e$ appears at most two
sets in $\mathcal{F}_\gamma$ by the definition.
Imagine that we replace $e$ into $e_1$ and $e_2$ for such sets. Recall that we consider the edges
on a line segment as one edge in the construction of the vertical decomposition.
Thus, the vertical decomposition of the subdivision induced by such a set remains the same.
Therefore, $\findcc$ also remains the same.

We now process the insertion of an edge $e$ by inserting a number of trapezoids to $\mathcal{T}$.
Here, we use $\Pi$ to denote the subdivision of complexity $n$ \emph{before} $e$ is inserted.
There are four cases: $e$ is not incident to $\lmm$, 
only one endpoint of $e$
is contained in $\lmm$, the endpoints of $e$ are contained in
distinct connected components of $\lmm$,
and the endpoints of $e$ are contained in the same connected component of $\lmm$. We can check if $e$ belongs to each case in $O(\log n)$ time
using the data structures described at the beginning of Section~\ref{sec:main}.
For the first three cases, we do not need to update $\mathcal{T}$.
In the first case, 
a new connected component, which consists of $e$ only, appears in the current subdivision.
However, the subdivision induced by the new connected component does not have any bounded face.
Therefore, we do not need to update $\mathcal{T}$.
In the second and third cases, no new face appears in the current subdivision.
We are required to update $\mathcal{T}$ only if 
$e$ makes a new face in the current subdivision. In other words,
the conditions on the definition of $\mathcal{F}_\gamma$ 
are not violated in these cases. (We will see this in more detail in Lemma~\ref{lem:correctness}.) 
Thus we do not need to update $\mathcal{T}$.

Now consider the remaining case: the endpoints of $e$ are contained in 
the same connected component, say $\gamma$, of $\lmm$. 
 Recall that $U(\gamma)$ is closed.
If $e$ is contained in the interior of $U(\gamma)$, we do nothing since $\mathcal{F}_\gamma$ covers $\gamma \cup e$.
Note that $e$ is contained in the interior of $U(\gamma)$ if and only if
 $e$ is contained in the interior of $U(\gamma\cup e)$.
 We can check in constant time if it is the case by Lemma~\ref{lem:in-union}. 
If $e$ is not contained in the interior of $U(\gamma)$, 
we trace the edges of the new face in time linear in the complexity of the face using the data structures presented at the beginning of Section~\ref{sec:main}. Then we compute the vertical decomposition of the face in the same time~\cite{triangulation}, and insert them to $\mathcal{T}$.
This takes time linear in the number of the new trapezoids inserted to $\mathcal{T}$, which is $O(n)$ in total over all updates by Lemma~\ref{lem:size-D1} and the fact that no trapezoid is removed from $\mathcal{T}$. 
As new trapezoids are inserted to $\mathcal{T}$, we update the 
stabbing-lowest data structure on $\mathcal{T}$. 

\begin{lemma}\label{lem:in-union}
	We can maintain a data structure of size $O(n)$ on $\mm$ supporting 
	$O(\log n)$ insertion time so that
	given an edge $e$ of $\mm$, we can check if it is contained in the interior of $U(\gamma)$
	in constant time, where $\gamma$ is the connected component of $\lmm$ containing $e$.
\end{lemma}
\begin{proof}
	We simply maintain a \emph{flag} for each edge of $\mm$ which has one of the three states: \emph{true}, \emph{false}, and \emph{null}.
	The flag of an edge $e$ is set to \emph{true} if and only if 
	$e$ lies on the boundary of $U(\gamma)$, where $\gamma$ is the connected component
	of $\lmm$ containing the edge. If the flag of an edge of $\mm$ is true,
	we give a direction to the edge so that each connected component of 
	the boundary of $U(\gamma)$ can be traversed in counterclockwise order around $U(\gamma)$.  
	The flag is set to \emph{null} if and only if it does not contained in  $U(\gamma)$.
	The flag is set to \emph{false} if and only if it is contained in the interior of $U(\gamma)$.
	Using the flag, we can check if an edge $e$ is contained in the interior of $U(\gamma)$ in constant time.
	
	We show that the flags can be maintained in $O(n\log n)$ total time
	in the course of the insertions of $n$ edges and vertices.
	A new vertex $v$ lying on an edge $e$ splits $e$ into two subedges.
	We find such two subedges in $O(\log n)$ time, and then we let each of them have the same flag as $e$ in constant time. We are done.

	Now consider the insertion of an edge $e$. 	
	We check if it is incident to $\lmm$ in $O(\log n)$ time. If not, we set the flag to \emph{null} since $U(e)$ is empty. 
	In the case that it is incident to $\lmm$, we find the edges incident to each endpoint $v$ of $e$ 
	that come before and after $e$ around $v$. 
	Using their flags and their directions, we can check in constant time
	if $e$ lies on the boundary of $U(\gamma)$. We set 
	the flag of $e$ accordingly. If the flag of $e$ is \emph{true}, 
	some edges of $\gamma$ 
	are required to 
	change their flags.  
	Such edges are the outer boundary edges of the new face made by the insertion of $e$. Moreover, if such an edge had a flag of \emph{null} (or \emph{true}), its flag
	are required to change to \emph{true} or \emph{false} (or \emph{false}).
	We trace the outer boundary of the new face from $e$ in time linear in the size of the outer boundary using the data structures presented at the  beginning of Section~\ref{sec:main}, and set the flag of each such edge to \emph{true} or \emph{false} accordingly.
	
	The total time for edge insertions
	is linear in the number of the total change on the flags due to the edge insertions and the time for checking if each edge is incident to $\lmm$.
	Since $\Pi$ is incremental, the flag value of an edge 
	turns to the true or false value only. Also, the false value does not turn
	to some other values. Therefore, the amount of the total change on the flags is $O(n)$, and the total update time is $O(n\log n)$.
\end{proof}
	
	For the correctness of the update algorithm, we have the following lemma.
\begin{lemma}
	For each connected component $\gamma$ of $\lmm$, there is a set $\mathcal{F}_\gamma$
	of connected subdivisions covering $\gamma$ such that 
	$\mathcal{T}_\gamma$ consists of 
	the cells of the vertical decompositions of the subdivisions of $\mathcal{F}_\gamma$ at any moment. \label{lem:correctness}
\end{lemma}
\begin{proof}
	Suppose that the lemma holds for every connected component of $\lmm$ 
	before $e$ is inserted, and then we are to show that
	the lemma holds after $e$ is inserted.	
	More specifically, we are to prove the following claim: 
	there is a set $\mathcal{F}_\gamma$ of connected subdivisions induced by
	edges of $\gamma$ such that
	(1) each edge of $\gamma$ is contained
	in at most two subdivisions, (2) one of the subdivisions contains 
	all edges of the boundary of $U(\gamma)$, and
	(3) $\mathcal{T}_\gamma$ consists of 
	the cells of the vertical decompositions
	of the subdivisions of $\mathcal{F}_\gamma$.
	For a connected component of $\gamma$ not incident to $e$, 
	we do not insert any trapezoid to $\mathcal{T}_\gamma$, 
	and $\gamma$ is still a maximal connected component of $\lmm$ after
	$e$ is inserted.
	Thus the claim still holds for such a connected component.
	In the following, we prove the claim for connected components incident to $e$.
	
	Consider the case that only one endpoint of $e$
	is contained in $\lmm$. Let $\gamma$ be 
	the connected component containing
	an endpoint of $e$. 
	In this case, no new face appears, and we do not insert any trapezoid to $\mathcal{T}$.
	We show that the claim still holds for the new connected component $\bar{\gamma}=\gamma\cup e$. By the assumption,
	$\mathcal{T}$ contains the cells of the vertical decompositions of the
	subdivisions in a set $\mathcal{F}_\gamma$ covering $\gamma$.  
	We just set  $\mathcal{F}_{\bar{\gamma}}$ to $\mathcal{F}_{{\gamma}}$.
	We show that $\mathcal{F}_{\bar{\gamma}}$ satisfies Conditions~(1--3).
	Since $\mathcal{T}$ remains the same and no new face appears in any of the subdivisions of $\mathcal{F}_{\bar{\gamma}}$,
	Condition~(3) holds immediately. 
	Since $e$ is contained in at most one subdivision of $\mathcal{F}_{\bar{\gamma}}$, Condition~(1) also holds.
	The boundary of $U(\bar{\gamma})$ coincides with the boundary of $U(\gamma)$. Therefore, Condition~(2) holds.
	
	Consider the case that two endpoints $e$ are contained in two distinct connected 
	components of $\lmm$. Let $\gamma$ and $\gamma'$ be such connected components.
	 In this case,
	no new face appears, and we do not insert any trapezoid to $\mathcal{T}$.
	We show that the claim still holds for the new connected component $\bar{\gamma}=\gamma\cup\gamma' \cup e$. By the assumption,
	$\mathcal{T}$ contains the cells of the vertical decompositions of the
	subdivisions in two sets
	 $\mathcal{F}_\gamma$ and $\mathcal{F}_{\gamma'}$ 
	covering $\gamma$ and $\gamma'$, respectively.  
	Since $\gamma$ and $\gamma'$ are different connected components of $\lmm$, there are three cases: $\gamma$ is contained in $U(\gamma')$,
	$\gamma'$ is contained in $U(\gamma)$, or $U(\gamma)$ and $U(\gamma')$ are
	disjoint. For the first case, $U(\bar{\gamma})$ coincides with $U(\gamma')$. 
	For the second case, $U(\bar{\gamma})$ coincides with $U(\gamma)$.
	For the third case, $U(\bar{\gamma})$ coincides
	with $U(\gamma)\cup U(\gamma')$ by definition.
	For the first and second cases, we let
	$\mathcal{F}_{\bar{\gamma}}$ be the union
		of $\mathcal{F}_\gamma$ and $\mathcal{F}_{\gamma'}$.
		Then Conditions (1--3) hold immediately. 
	For the last case, we consider the two subdivisions, say $\Pi_\gamma'$ and $\Pi_{\gamma'}'$, 
	from $\mathcal{F}_\gamma$ and $\mathcal{F}_{\gamma'}$ containing
	the boundary edges of $U(\gamma)$ and $U(\gamma')$, respectively.
	We merge two subdivisions, and insert $e$ to the resulting subdivision. 
	Let $\Pi_{\bar{\gamma}}'$ be 
	the resulting subdivision.
	Notice that it is connected.
	We let 	$\mathcal{F}_{\bar{\gamma}}$ be the union of
	 $\mathcal{F}_\gamma$ and $\mathcal{F}_{\gamma'}$ excluding  $\Pi_\gamma'$ and $\Pi_{\gamma'}'$ and including $\Pi_{\bar{\gamma}}'$.
	We show that $\mathcal{F}_{\bar{\gamma}}$ satisfies Conditions~(1--3).
	Since $\mathcal{T}$ remains the same and 
	the set of the cells of the vertical decompositions of $\Pi_\gamma'$ and $\Pi_{\gamma'}'$ coincides with the set of the vertical decomposition of $\Pi_{\bar{\gamma}}'$, Condition~(3) holds. 
	Since $\gamma$ and $\gamma'$ are different connected components of $\lmm$, Condition~(1) also holds immediately.
	Also, every boundary edge of $U(\bar{\gamma})$ 
	appears on the boundary 
	of exactly one of $U(\gamma)$ and $U(\gamma')$. Therefore, Condition~(2) also holds.
	
	Now consider the case that both endpoints of $e$ are contained in the 
	same connected component of $\lmm$, say $\gamma$.
	In this case, a new face $F$ appears on the subdivision induced by $\bar{\gamma}=\gamma\cup e$. Note that $F$ has no hole since $\bar{\gamma}$ is connected. If $F$ is contained in the interior of $U(\gamma)$, we do nothing
	and set $\mathcal{F}_{\bar{\gamma}}$ to $\mathcal{F}_\gamma$. 
	Conditions (1--3) hold immediately.
	Now assume that $F$ is not contained in the interior of $U(\gamma)$.
	By construction, 
	we insert the cells of the vertical decomposition of $F$ to $\mathcal{T}$. 
	By the assumption, we have $\mathcal{F}_\gamma$ satisfying Conditions~(1--3).
	One of the subdivisions of $\mathcal{F}_\gamma$ contains the edges on the  boundary of $U(\gamma)$. Imagine that we add the edges on the boundary of $F$
	to such a subdivision of $\mathcal{F}_\gamma$.
	This subdivision remains to be connected since
	the boundary of $F$ is incident to the boundary of $U(\gamma)$. 
	Conditions~(1--3) hold because an edge on the boundary of $F$ appearing on $U(\gamma)$
	lies in the closure of $U(\bar{\gamma})$. 
	Therefore, the claim holds for any case after $e$ is inserted.
\end{proof}

Notice that we do not remove any trapezoid from $\mathcal{T}$. 
Let $S(n), Q(n)$ and $U(n)$ be the size, the query time and the update time
of an insertion-only stabbing-lowest data structure for $n$ trapezoids, respectively. In the case of the data structure described in Section~\ref{sec:stabbing}, we have $S(n)=O(n\log n)$, $Q(n)=O(\log^2 n)$
and $U(n)=O(\log n\log\log n)$.
By Lemma~\ref{lem:size-D1},
the total number of trapezoids inserted to $\mathcal{T}$ is $O(n)$.
Then we have the following lemmas.

\begin{lemma}
	The total update time for $n$ insertions of edges and vertices 
	is $O(n\cdot U(n))$ time.
\end{lemma}

\begin{theorem}
	We can construct a data structure of size $O(S(n))$ so that
	the connected component of $\lmm$ containing the outer boundary
	of the face containing $q$ can be found in $O(Q(n))$ worst case time
	for any point $q$ in the plane,
	where $n$ is the number of edges at the moment.
	Each update takes $O(U(n))$ amortized time.
\end{theorem}

\subsection{$\locatecc(\gamma)$: Find the Face Containing a Query Point in $\Pi_\gamma$}\label{sec:D2}
For each connected component $\gamma$ of $\lmm$, we maintain a data structure, which is denoted by $\locatecc(\gamma)$, for finding the face of 
$\Pi_\gamma$ containing a query point. Here, we need two update operations for $\locatecc(\cdot)$:
inserting a new edge to $\locatecc(\cdot)$ and merging two data structures $\locatecc(\gamma_1)$ and $\locatecc(\gamma_2)$ for two connected components $\gamma_1$ and $\gamma_2$ of $\lmm$. Notice that we 
do not need to support edge deletion since $\mm$ is incremental.

No known point location data structure supports the merging operation explicitly. 
Instead, one simple way is to make use of the edge insertion operation which is  supported by most of the known data structures for the dynamic point location problem. For example, we can use a point location data structure for incremental subdivisions
given by Arge et al.~\cite{ABG-Improved-2006}. 
Its query time is $O(\log n\log^* n)$ and amortized insertion time
is $O(\log n)$ under the pointer machine model.
For merging two data structures, we simply insert every
edge in the connected component of smaller size to the data structure
for the other connected component.
By using a simple charging argument, we can show that 
the total update (insertion and merging) takes $O(n\log^{2}n)$ time.
The query time is $O(\log n\log^* n)$. 

In this section, we improve the update time 
at the expense of increasing the query time.
Because $\findcc$ requires $O(\log^2 n)$ query time, we are allowed to spend more time on a point location query on a connected component of $\Pi$. 
The data structure proposed in this section supports
$O(\log^2 n)$ query time. The total update time (insertion and merging)
is $O(n\log n\log\log n)$.

\subsubsection{Data Structure and Query Algorithm}
$\locatecc(\gamma)$ allows us to find the face of $\Pi_\gamma$ containing
a query point. Since $\gamma$ is connected and we maintain the outer boundary of each face of $\mm$, 
it suffices to construct a vertical ray shooting structure for the edges of $\gamma$. Recall that the boundary of a face of $\Pi_\gamma$ coincides
with the outer boundary of a face of $\mm$.  
The vertical ray shooting problem is \emph{decomposable} in the sense that
we can answer a query on $\mathcal{S}_1\cup\mathcal{S}_2$ in constant time
once we have the answers to queries on $\mathcal{S}_1$ and $\mathcal{S}_2$ for any two sets $\mathcal{S}_1$ and $\mathcal{S}_2$ of line segments in the plane.
Thus we can use an approach by Bentley and Saxe~\cite{decomposable}.

We decompose the edge set of $\gamma$ into
subsets of distinct sizes
such that each subset consists of exactly $2^i$ edges
for some index $i\leq \lceil \log n\rceil$. Note that there are $O(\log n)$
subsets in the decomposition. We use $\bb(\gamma)$ to denote
the set of such subsets, and $\bb$ to denote the union of $\bb(\gamma)$
for all connected components $\gamma$ of $\lmm$. 
$\locatecc(\gamma)$ 
consists of $O(\log n)$ static vertical ray shooting data structures, one for 
each subset in $\bb(\gamma)$. To answer a query on $\gamma$, we apply a vertical ray shooting query on each subset of $\bb(\gamma)$, and choose the one lying immediately above the query point. This takes $O(Q_s(n)\log n)$ time, where $Q_s(n)$ denotes the query time of the static vertical ray shooting data structure we use.
For a static vertical ray shooting data structure, we present a variant of the (dynamic) vertical ray shooting data structure of 
Arge et al.~\cite{ABG-Improved-2006} because it can be constructed in 
$O(N\log\log n)$ time, where $N$ is the number of the edges in the data structure,
once we maintain an auxiliary data structure, which we call the \emph{backbone} tree.

\paragraph{Backbone tree.} It is a global data structure constructed on 
all edges of $\Pi$ while $\locatecc(\gamma)$ is constructed on the edges
of each connected component $\gamma$ of $\lmm$. The backbone tree allows us 
to construct a static vertical ray shooting data structure $\ds(\beta)$ in 
$O(|\beta|\log\log n)$ time for any subset $\beta$ of $\bb$.

The backbone tree consists of two levels. The base tree is 
an interval tree of fan-out $\log^\epsilon n$ on the edges of $\Pi$
for an arbitrary fixed constant $0<\epsilon<1$. 
The height of the base tree is $O(\log n/\log\log n)$. 
For the definition and notations
for the interval tree, refer to Section~\ref{sec:preliminary}. 
Each node $v$ has three sets of edges of $\Pi$: $\lseg(v)$, $\mseg(v)$,
and $\rseg(v)$. 
For each node $v$, the edges of $\lseg(v)$ (and $\rseg(v)$) have their endpoints on a common
vertical line. Thus they can be sorted in $y$-order.
Recall that we decompose the edge set of each connected component into $O(\log n)$ subsets of distinct sizes, and we denote the set of such subsets
by $\bb$. For each subset $\beta$ of $\bb$, we maintain the sorted list of 
the edges of $\lseg(v)$ (and $\rseg(v)$) contained in $\beta$
 with respect to their $y$-order. 

Also, for each node $v$, the edges of $\mseg(v)$ have their endpoints on
$O(\log^\epsilon n)$ vertical lines. 
This is because we store the part of $e$ excluding the union of $\region{v_1}$ and $\region{v_2}$ to $\mseg(v)$, 
where $v_1$ and $v_2$ are the children of $v$
such that $\region{\cdot}$ contains the endpoints of $e$.  
We construct a segment tree on the edges of $\mseg(v)$ (with respect to the $x$-axis). Note that the segment tree has height of $O(\log\log n)$.
In the segment tree associated with $v$,
each edge of $\mseg(v)$ is stored in $O(\log\log n)$ nodes.
Then for each node $u$ of the segment tree, every edge stored in the node 
crosses the left and right vertical lines on the boundary of 
$\region{u}$, and thus they can be sorted with respect to the $y$-axis.
For each subset $\beta$ of $\bb$, we maintain the sorted list of the edges of $\beta$ stored in $u$ 
with respect to the $y$-axis. Then we have a number of sorted lists, each for 
a subset of $\bb$.

The size of the backbone tree is $O(n\log\log n)$. To see this, observe that each
edge of $\mm$ is stored in at most
three nodes of the base tree, and $O(\log\log n)$ nodes in
secondary trees associated with nodes of the base tree. 

\paragraph{Contracted backbone tree for a subset of $\bb$.}
The backbone tree contains all edges of $\Pi$ and has size of $O(n\log\log n)$.
We are to extract the information of a subset $\beta$ of $\bb$ 
from the backbone tree and construct a tree of size $O(|\beta|\log\log n)$ as follows. Here, we maintain this tree for every subset $\beta$ of $\bb$
as well as the backbone tree. 
Each edge of $\beta$ is stored in at most three nodes in (the base tree of) the backbone tree. 
Let $V$ be the set of the nodes $v$ of 
the backbone tree such that $\lseg(v)$, $\rseg(v)$ or $\mseg(v)$
contains an edge of $\beta$. 
Imagine that we remove a subtree of the base tree if no node of the subtree is in $V$. Also, we imagine that
we contract a node of the base tree if it has only one child. That is,
we remove this node and connect its parent and its child by an edge.
Note that the resulting tree is not necessarily balanced, 
but its height is $O(\log n/\log\log n)$. 
We call a node of the resulting tree which is not in $V$ a \emph{dummy node}.
Note that the number of the dummy nodes is $O(|V|)$. 
For a non-dummy node, we store a pointer pointing to its corresponding node
in the backbone tree.
The resulting tree is the base tree of the contracted backbone tree for $\beta$.

As secondary structures, each node $v$ of the base tree of the backbone tree
has several sorted lists of edges. Among them, we choose the sorted lists 
of edges of $\beta$ only. More specifically, we have the sorted list 
of the edges of $\beta$ in $\lseg(v)$ (in the backbone tree) with respect to their $y$-order.
Similarly, we have the sorted list of the edges of $\beta$ in $\rseg(v)$ 
with respect to their $y$-order.
We choose them and associate them with $v$ in the contracted backbone tree. 
For the edges in $\mseg(v)$, the node $v$ has an associated 
segment tree $T$ of height $O(\log\log n)$.  
Each node of the segment tree has at most one sorted list for edges of $\beta$.
For the contracted backbone tree for $\beta$, we choose the nodes of $T$
which have sorted lists for $\beta$. Let $V_T$ be the set of such nodes.
 Then we remove a subtree of the segment tree 
if no node in the subtree is in $V_T$. But unlike the base tree,
we do not contract a node even though it has only one child.
This makes the merging procedure efficient.
 Then for a node of $V_T$ in the remaining tree, we associate
the sorted list for $\beta$ stored in the node in the backbone tree with the node.
These lists are the secondary and tertiary structures  of the contracted backbone tree for $\beta$.

Now we analyze the space complexity of the contracted backbone tree. 
Each edge of $\beta$ is stored in three nodes in the base tree of the contracted backbone tree, and it is stored in $O(\log\log n)$ nodes of the segment tree associated with a node of the base tree.
Moreover, the size of the base tree is linear in the size of $\beta$,
and the size of the segment tree associated with a node $v$ of the base tree is 
$O(m\log\log n)$, where $m$ is the number of edges of $\mseg(v)$ contained in $\beta$. Here we have an $O(\log\log n)$ factor because we allow a node has only one child in the segment trees unlike the base tree, and because  
the height of a segment tree is $O(\log\log n)$.
Therefore, the total size of the contracted backbone tree for $\beta$
is $O(|\beta|\log \log n)$. 

\begin{lemma}\label{lem:const-contracted-tree}
	Assume that two subsets $\beta_1$ and $\beta_2$ of $\bb$ are merged
	into a subset $\beta$. 
	Given two contracted backbone trees for 
	$\beta_1$ and $\beta_2$, 
	we can update the backbone tree  and construct the contracted
	backbone tree for $\beta$ in $O(|\beta|\log\log n)$ time.
\end{lemma}
\begin{proof}
	Let $T_1$ and $T_2$ be the contracted backbone trees 
	for $\beta_1$ and $\beta_2$, respectively. 
	Let $T$ be the contracted backbone tree for 
	$\beta$.
	Every non-dummy node $v$ of $T$ is a
	non-dummy node of $T_1$ or $T_2$ by definition. 
	We first compute the non-dummy nodes of $T$ and sort them
	in the order specified by
	the pre-order traversal of $T$. 
	To do this, we apply the pre-order traversal on $T_i$, and sort the nodes of $T_i$ in this order in $O(|\beta_i|)$ time for $i=1,2$.
	Then we merge two sorted lists in $O(|\beta_1|+|\beta_2|)$.
	Here, we can check for two nodes $v_1$ and $v_2$, 
	one from $T_1$ and one from $T_2$,
	if $v_1$ comes before $v_2$ in the sorted list for $T$ in constant time
	using $\region{v_1}$ and $\region{v_2}$. Recall that each non-dummy node of a contracted backbone tree points to its corresponding node in the backbone tree.
	Let $L$ be the merged list.
	
	Then we compute dummy nodes of $T$ and put them in $L$. 
	Using $\region{\cdot}$, we construct a tree $T'$ such
	that the pre-order traversal of $T'$ gives $L$ and $\region{u}$
	contains $\region{v}$ for an ancestor $u$ of a node $v$ in $T'$, 
	which is unique.
	Notice that a fan-out of each node of $T'$ is not necessarily
	at most $\log^{\epsilon} n$. 
	Thus we again consider each node of $T'$ one by one. If 
	a node $v$ of $T'$ has more than $\log^{\epsilon} n$ children, we insert
	dummy nodes as descendants of $v$ so that $v$ and
	the new dummy nodes have at most
	$\log^{\epsilon} n$ children. More specifically, we construct
	a balanced search tree $T'_v$ of fan-out $\log^\epsilon n$ on the node $v$ and its children $u_1,\ldots,u_k$ such that the root node is $v$ and each $u_i$ corresponds to a leaf node of $T'_v$. Then we replace the edges connecting $v$ and 
	its children in $T'$ with $T'_v$.
	The maximum fan-out of a node of $T'$ is $\log^{\epsilon} n$ and
	 the height of $T'$ is $O(\log n/\log\log n)$.
	 In this way, we have the base tree of the contracted backbone tree for $\beta$.

	Now we construct the secondary structure for a node $v$.
	Since a dummy node does not have a secondary structure, we assume that 
	$v$ is a non-dummy node of $T$. 
	Since $v$ is a non-dummy node in $T_i$ for $i=1,2$, 
	we have the pointer pointing to its corresponding node in the backbone tree. 
	If $v$ appears only one of  $T_1$ and $T_2$, we  
	simply copy the secondary structure of $v$ in $T_1$ or $T_2$.
	If $v$ appears in both $T_1$ and $T_2$,
	we are required to merge the secondary structures of $v$ in the backbone tree.
	
	Specifically, consider the secondary structure
	on $\lseg(v)$ (and $\rseg(v)$) for $T_i$.
	It is the sorted list of the edges in $\lseg(v)$ (and $\rseg(v)$) contained in
	$\beta_i$ 
	with respect to the $y$-order. 
	We can merge two sorted lists in time linear in the total size of the sorted lists. We merge them in the backbone tree, and copy the resulting list to the contracted backbone tree $T$. 
	The secondary structure on $\mseg(v)$ is a binary search of height $O(\log\log n)$ such that a sorted list 
	of edges with respect to the $y$-order is associated with each node. 
	As we did for the base tree,
	we merge two segment trees, one from $T_1$ and one from $T_2$, in time linear
	in the total complexity of the two segment trees. 
	Then we merge the sorted lists stored in each node of the segment tree,
	and then copy it to $T$.
	Therefore, we can update the backbone tree and 
	obtain all secondary structures of $T$
	in time linear in the total size, which is $O(|\beta|\log\log n)$.
\end{proof}

\paragraph{Static vertical ray shooting data structure.}
For each subset $\beta$ of $\bb$, we maintain a static data structure $\ds(\beta)$
 for answering a vertical ray shooting query in $O(\log n)$ time. 
 This is a variant of the vertical ray shooting data strucutre of Arge et al.~\cite{ABG-Improved-2006}. 
It is an interval tree with fan-out $f=\log^\epsilon n$ for an arbitrary fixed constant $0<\epsilon<1$ such that segment trees and priority
search trees are associated with each node as secondary structures.  
Here, we only provide the description of the data structure together with
the query algorithm. Its construction and merge procedure
will be described at the end of this subsection.

Its base tree is the contracted backbone tree for $\beta$.
We associate several secondary structures with each node of the base tree.
For a query point $q$, we walk along the search path $\pi$ of $q$ in the base tree. 
We consider secondary structures 
of each node on $\pi$. 
Every edge of $\beta$ intersecting the vertical line containing $q$ is stored
in a secondary structure of a node in $\pi$. 
We find the edge lying immediately above the query point among the edges 
of each of $\lseg(v)$,
$\rseg(v)$ and $\mseg(v)$ in $O(\log\log n)$ time for each node $v$ in $\pi$,
which leads to $O(\log n)$ query time.
Then we choose the one lying immediately above the query point.

Consider the pieces in $\mseg(v)$.
A segment tree on $\mseg(v)$ of height $O(\log\log n)$ is associated with $v$, 
and thus we can find the piece in $\mseg(v)$ lying immediately above a query point.
 Without fractional cascading, we are required to
spend $O(\log n)$ time on each node of the segment tree for applying binary search
on the sorted lists (assuming that the sorted lists are maintained in balanced binary search trees), which leads to
$O(\log^2 n)$ total query time on all $\mseg(\cdot)$'s along $\pi$.
To improve it, we use factional cascading so that a binary search
on each node of a segment tree can be done in constant time after the initial binary search performed once along $\pi$.
We show how to apply fractional cascading in this case at the end of the description of the whole structure of $\ds(\beta)$.
Then the query time on $\mseg(\cdot)$'s along $\pi$ is $O(\log n)$ in total.

Now consider the pieces in $\lseg(v)$. Recall that we store the part of an edge
lying inside $\region{v}$ in $\lseg(v)$. Thus 
the right endpoints of the pieces are on a common vertical line. 
Thus we can use a priority search tree on $\lseg(v)$ to find
the piece of $\lseg(v)$ lying immediately above a query point in $O(\log n)$ time, which leads
to the total query time of $O(\log^2n/\log\log n)$ on $\lseg(\cdot)$'s. To improve the query time on each node to $O(\log\log n)$,
we partition the pieces in $\lseg(v)$ into $O(|\lseg(v)|/\log^2 n)$ blocks 
with respect to their $y$-order such that each block
consists of $O(\log^2 n)$ pieces. We construct the priority search tree
on each block so that we can find the piece lying immediately above
a query point in $O(\log\log n)$ time among the pieces in each block. 
Also for each block, we find the piece with the leftmost left endpoint
and we call it the \emph{winner} of this block.
We store the winner of each block in the nodes of the path from $v$
to the leaf node corresponding to the left endpoint of the winner
(including $v$) in the base tree. When we store a winner to a node $u$,
we indeed store the part of the winner lying outside of $\region{u'}$
for the child $u'$ of $u$ such that
$\region{u'}$ contains the left endpoint of the winner. Let $\mathcal{L}_w(u)$ be
the set of pieces of the winners stored in a node $u$ in the base tree.
In this way, the endpoints of the pieces contained in $\winner(u)$ for 
a node $u$ has $O(\log\log n)$ distinct $x$-coordinates.
Then we construct the segment tree on $\winner(u)$  of height $O(\log\log n)$
for each node $u$ of the base tree.

In summary, we have two substructures with respect to $\lseg(\cdot)$ 
for each node $v$:
$O(|\lseg(v)|/\log^2 n)$ priority search trees for $\lseg(v)$ and one segment tree
of height $O(\log\log n)$ for all winners of the blocks of $\lseg(\cdot)$ for 
all nodes in the path from $v$ to the root of the base tree. 
Given a query point, we are to find the edge of $\cup_{v\in \pi} \lseg(v)$ 
lying immediately above the query point.
To do this, we search the segment trees for the winners associated with the nodes of $\pi$.
Then we find the winner $w$ lying immediately above $q$ among them
in $O(\log n)$ time in total
with fractional cascading.
Note that  the edge of $\cup_{v\in \pi} \lseg(v)$ lying immediately above $q$ 
is not necessarily a winner of some block. Thus we are also required to search
priority search trees. Here, it is sufficient to search the priority
search tree for only one block for each node of $\pi$.
To be specific, for each node $v$ of $\pi$, we find the winner $w_v$ lying immediately
above the right endpoint of $w$ if it is contained in $\region{v}$, or $w\cap \ell(v)$, otherwise,  among all winners of the blocks for $v$.
An edge lying above $q$ but below $w$ is 
contained in the block containing $w_v$ if such an edge exists in $\lseg(v)$.
Thus it is sufficient to search the priority search tree for this block.
In the construction of $\ds(\cdot)$, we compute $w_v$ 
 in advance 
for each winner $w$ and each node $v$ on the path from the root to the
leaf node corresponding to the left endpoint of $w$. Then we can
choose the priority search trees to search further in time linear in
the number of such priority search trees, which is $O(\log n/\log\log n)$.
In this way, we can find the edge of $\cup_{v\in \pi} \lseg(v)$ 
lying immediately above the query point in $O(\log n)$ time.

We can construct data structures of $\mathcal{R}(v)$ analogously.
Therefore, we have a static data structure supporting $O(\log n)$ query time.

\paragraph{Fractional cascading on $\mseg(\cdot)$.}
Fractional cascading is a technique that allows us to apply binary searches on lists associated with edges on a path of a graph efficiently~\cite{fractional}.
In our case, the underlying graph $G$ is a binary tree each of whose node
$v$ has a list $L(v)$.
Our goal is to find the predecessor of a query $q'$ in $L(v)$ for all nodes
$v$ in the path from the root to a given leaf node of $G$ in $O(p+\log N)$ time in total, where $p$ is the length of the path and $N$ is the total number of the elements in $L(\cdot)$'s.
Here, we use fractional cascading for the segment trees on
the edges of $\mseg(\cdot)$ and the segment trees on the winners of $\lseg(\cdot)$.
We show how to apply this for $\mseg(\cdot)$ only. We can do
this for $\lseg(\cdot)$ analogously.

Fractional cascading introduced by Chazelle and Guibas~\cite{fractional} works as follows.
We first assume that every element in $L(\cdot)$'s comes from an ordered set, for example, $\mathbb{R}$.
 Starting from the root $v$ of $G$, we choose every fourth element, and insert them to the lists of its children. We give a pointer to each element in $L(v)$ to the same element in $L(u)$ if it exists, for a node $v$ and its child $u$.
Also, we let each element in $L(u)$ point to
its predecessor among the elements in $L(u)$ which come from the list of its parent.
In this way, once we apply binary search on a leaf node, we can just follow
the pointers to compute the predecessor of $L(v)$ for every node $v$ in the 
path from the left node to the root. The number of the elements of 
the lists $L(\cdot)$ remains the same asymptotically.

In our case, the contracted backbone tree of $\beta$ has $\log^{\epsilon} n$ fan-out, and it is a two-level structure. Instead of considering the two-level tree,
we consider a binary search tree of height $O(\log n)$ obtained
by linking the secondary structures in a specific way. 
Imagine that we remove all edges of the base tree of the contracted backbone tree. We will connect a node $v$ of the base tree with the leaf $u$ of the secondary segment tree associated with its parent with $\region{u}\subseteq \region{v}$.
Even though $u$ and $v$ are nodes of different trees (the base tree and a segment tree), either $\region{u}\subseteq \region{v}$ or $\region{u}\cap \region{v}=\emptyset$ by construction.
In this way, the resulting graph $G$ forms a tree, but the degree of a node
might be more than two in the case that more than two children of a node $v$ of the base tree are connected to the same node in the segment tree associated with $v$.
In this case, we simply make this part be a balanced binary search tree by adding a set of dummy nodes of size linear in the degree of $v$ as we did
in Lemma~\ref{lem:const-contracted-tree} for making a tree have smaller fan-out.
Recall that the base tree of the contracted backbone tree has height $O(\log n/\log\log n)$ and the segment tree associated with a node of the base tree has
height $O(\log\log n)$.
Therefore $G$ is a tree of height $O(\log n)$ and of size $O(n)$.

Then the nodes in the base tree and the segment trees of the contracted backbone tree
we visited are in a single path $\pi'$ from the root to a leaf node of $G$. 
We can find the path in $O(\log n)$ time.
Recall that each node $v$ of $G$ has $\region{v}$ (defined on the segment tree from  which $v$ comes). Also, it has a list $L(v)$ of edges of $\beta$ crossing $\region{v}$ sorted with respect to their $y$-order.
We are to find the predecessor of a query point in $L(v)$ for all nodes $v$ in $\pi'$.
Here, one difficulty is that it is not possible to globally $y$-order all edges of $\mm$. However, using the fact that the edges of $L(v)$ for all nodes $v$ in a single path from the root to a leaf node can be globally $y$-ordered,
we can apply fractional cascading on $G$. In fact, several previous work~\cite{ABG-Improved-2006,BJM-Dynamic-1994,cn-locatoin-2015} uses this observation to apply fractional cascading on a segment tree.

Therefore, we have the following lemma. Here, notice that a query time on $\ds(\beta)$ for each subset $\beta\in\bb$ is $O(\log n)$. 
\begin{lemma}
	Given the data structure $\ds(\beta)$ for every subset $\beta\in\mathcal{B}$,
	we can find the edge lying immediately above a query point among the edges
	of a connected component $\gamma$ of $\lmm$ in $O(\log^2 n)$ time.
\end{lemma}

\paragraph{An $O(|\beta|\log\log n)$-time merge operation on two static data structures.}
Suppose that we are given two static vertical ray shooting data structures
$\ds(\beta_1)$ and $\ds(\beta_2)$ for two subsets $\beta_1$ and $\beta_2$ of $\bb$. We are to construct $\ds(\beta)$ 
in $O(|\beta|\log\log n)$ time, where $\beta=\beta_1\cup\beta_2$. 
Since we can construct 
the contracted backbone tree for $\beta$ in $O(|\beta|\log\log n)$ time by Lemma~\ref{lem:const-contracted-tree}, the remaining task is 
to construct the secondary structures for $\mseg(v)$, $\lseg(v)$ and $\rseg(v)$ for each node $v$ in the base tree. 

For each node $v$, we first consider 
the secondary structure for $\mseg(v)$. In fact, we can answer
a vertical ray shooting query on $\mseg(v)$ using the segment tree associated with $v$ in the base tree of the contracted backbone tree. But it does not support fractional cascading. Thus we are required to 
construct the structure for fractional cascading in the base tree and all segment trees. As we mentioned above, we construct a binary tree $G$ of height $O(\log n)$ consisting all nodes of the base tree and all segment trees. This takes time linear in the size of $\ds(\beta)$, which is $O(|\beta|\log\log n)$.
Then we construct the structure for fractional cascading on $G$ 
 in time linear
in $G$~\cite{fractional}, which is $O(|\beta|\log \log n)$ in this case.

Now consider the secondary structure for $\lseg(v)$. We have
the sorted list of the edges
of $\lseg(v)$ with respect to 
their $y$-order.
We partition the edges of $\lseg(v)$ with respect to their $y$-order such that
each block contains $O(\log^2 n)$ edges in time linear in
the size of $\lseg(v)$. In the same time, we can obtain
the sort lists of the edges in each block with respect to 
their $y$-order.
We construct a priority search tree for each block in time linear
in its size using the sorted list.
Then we choose the winner of each block in time linear in the block size. Notice that there are $O(|\beta|/\log^2 n)$ winners in total for 
all nodes of the base tree of the contracted backbone tree.
We store each winner to the nodes in the path of the base tree
from the node defining it to the leaf corresponding to its left endpoint. And for each node $v$ of the base tree, we construct
a segment tree on the winners stored in $v$ in $O(k_v\log k_v)$,
where $k_v$ is the number of the winners stored in $v$.
Since the total number of winners is $O(|\beta|/\log^2 n)$ and the 
height of the base tree is $O(\log n/\log\log n)$, the sum
of $k_v$ for every node $v$ is $O(|\beta|/(\log n\log\log n))$.
Therefore, we can construct the segment trees on the winners
for each node in total $O(|\beta|/\log\log n)$ time.

We also compute the winner $w_v$ lying immediately above
$w$ (more precisely, above $\ell(v)\cap w$ or above the right endpoint of $w$) among all winners of the blocks for $v$ in advance 
for each winner $w$ and each node $v$ on the path from the root to the
leaf node corresponding to the left endpoint of $w'$. 
To do this, we compute the lower envelope of the winners of the blocks in the same node in $O(|\beta|/\log n)$ time in total for every node in the base tree.
Then for each winner $w$, we walk along the base tree from the leaf node corresponding to the left endpoint of $w$ to the root, and compute the winner
lying immediately above $w$ in the lower envelope stored in each such node. For each winner $w$, this takes $O(\log^2 n/\log\log n)$ time.
Since there are $O(|\beta|/\log^2 n)$ winners, the running time is $O(|\beta|/\log\log n)$ in total.

\begin{lemma}
	Given $\ds(\beta_1)$ and $\ds(\beta_2)$ for two subsets $\beta_1$ and $\beta_2$ of $\bb$, we can construct $\ds(\beta)$
	in $O(|\beta|\log\log n)$ time, where $\beta=\beta_1\cup\beta_2$.
\end{lemma}

\subsubsection{Update Procedure (without Rebalancing)}
In this subsection, we present an update procedure for $\locatecc(\cdot)$.
We have two update operations, the insertion of edges and vertices.
Here, we do not update $\locatecc(\cdot)$ in the case of the insertion of a vertex.
In other words, we treat edges whose union forms a line segment as one edge.
In this way, it is possible that 
the answer to a query on $\locatecc(\cdot)$ returns a line segment 
$\ell$ containing the edge
we want to find. To find the edge on $\ell$, we maintain the union of the edges for every such set of edges, and maintain the sequence of the edges on the union in order. Then after finding $\ell$, we apply binary search on the edges on $\ell$ to find the solution. 

Now suppose that we are given an edge $e$ and we are to update
$\locatecc(\cdot)$ accordingly. Specifically, we update
the static vertical ray shooting data structures for some
subsets of $\bb$ and the backbone tree.
For simplicity, we 
first assume that the edges
to be inserted are known in advance so that we can keep the base tree and segment
trees balanced. At the end of this subsection, we show how to get rid of 
this assumption and show how to rebalance the trees without increasing the update time.
Here, we use $\Pi$ to denote the subdivision of complexity $n$ \emph{before} $e$ is inserted.

We find the connected components of $\lmm$ incident to $e$ in $O(\log n)$ time. There are three cases: there is no such connected component,
there is only one such connected component, or 
there are two such connected components. 

\paragraph{Case 1. No connected component is incident to $e$.}
In this case, a new connected component consisting of only one edge $e$ appears.
We make a new subset $\beta$ consisting of only one edge $e$ and insert it to 
$\bb$.
Then we update the backbone tree by inserting $e$ as follows.

We find the list $\mseg(v)$ where $e$ is to be inserted for a node $v$ in the base in $O(\log n)$ time.
It is the 
lowest common ancestor of two leaf nodes corresponding to the
two endpoints of $e$. 
The node $v$ has at most $f$ children, say $u_1,\ldots,u_{f'}$. 
By construction, one endpoint is in $\region{u_j}$
and the other endpoint is in $\region{u_k}$ for two indices 
$1\leq j < k \leq f$.
We split $e$ into three pieces: $e\cap \region{u_j}$, $e\cap \region{u_k}$,
and the other piece in $O(\log\log n)$ time.

We first update the segment tree constructed on $\mseg(v)$ by inserting
the piece $e_m$ of $e$ lying outside of $\region{u_j}\cup\region{u_k}$. There are $O(\log\log n)$ nodes $w$ in the segment
tree such that $\region{w}\subseteq e_m$ and $\region{w'}\not\subseteq e_m$ for the parent $w'$ of $w$. 
We associate the list consisting of only one edge $e$ with each such node. Recall that $\beta$ consists
of $e$ only. In this way, we complete the update 
procedure for the segment tree on $\mseg(v)$ in $O(\log\log n)$ time.
Then we update the secondary structure of $\lseg(v)$. It consists of
the sorted list of the edges in each subset of $\bb$ with respect to their $y$-order. Since $\beta$ consists only one edge, we make the sorted list containing $e$ only, and associate it with $v$.
We also do this for $\rseg(v)$. In this way, we have the backbone tree
of the current subdivision after $e$ is inserted.

Then we construct $\ds(\beta)$ in $O(\log\log n)$ time. The base tree
of $\ds(\beta)$ consists of only three nodes: the root $v$ and its
children $u_j$ and $u_k$. The root $v$ has a segment tree
 storing only one edge, but the edge is stored in $O(\log\log n)$ nodes.
In other words, the segment tree associated with $v$ is a path consisting of $O(\log\log n)$ nodes.
For each of the two nodes $u_j$ and $u_k$ of the base tree, we construct a priority
search tree for $e$ of constant size. Then we are done.

\paragraph{Case 2. Only one connected component is incident to $e$.}
If there is only one connected component, say $\gamma$, we 
update $\ds(\gamma)$. Recall that we have $\bb(\gamma)$, which
is a set of $O(\log n)$
subsets of the edge set of $\gamma$ 
of distinct sizes. We make a new subset consisting only one
edge $e$ and add it to $\bb(\gamma)$. And we update the backbone tree
as we do in Case~1.

If there is another subset of $\bb(\gamma)$ of size $1$, we merge
them into a new subset of size $2$. Then we merge their static vertical
ray shooting data structures and update the backbone tree as well
 using Lemma~\ref{lem:const-contracted-tree}.
We do this until every subset of $\bb(\gamma)$ has distinct size. 
Then we are done.

\paragraph{Case 3. Two connected components are incident to $e$.}
If there are two such connected components, say $\gamma_1$ and $\gamma_2$, they are merged into one connected component together
with $e$.  If every subset in $\bb(\gamma_1)$ and $\bb(\gamma_2)$ 
has distinct size, we just collect every static vertical ray shooting
data structure constructed on a subset in $\bb(\gamma_1)\cup\bb(\gamma_2)$, and 
 we are done. 
If not, we first choose the largest subsets, one from $\bb(\gamma_1)$ and one from $\bb(\gamma_2)$, of the same size, say $2^i$. Then we merge them in the 
backbone tree and construct a new vertical ray shooting data structure on the union $\beta'$ of the two subsets in $O(2^{i+1}\log\log n)$ time. If there is a subset in $\bb(\gamma_1)$ or $\bb(\gamma_2)$ of size 
 $2^{i+1}$ other than $\beta'$, we again merge them together to form a subset of size $2^{i+2}$. We repeat this until every subset in $\bb(\gamma_1)$ and $\bb(\gamma_2)$ of size at least $2^i$ has distinct size.
 Then we consider the largest subsets, one from $\bb(\gamma_1)$ and 
 one from $\bb(\gamma_2)$, of the same size again. Note that the size of the two subsets is less than $2^i$.
 We merge them, and repeat the merge procedure.
 We do this for every pair of subsets in $\bb(\gamma_1)$ and $\bb(\gamma_2)$ of the same size. Finally, we have the set $\bb(\gamma_1\cup\gamma_2)$ of 
 of subsets of the edges of $\gamma_1\cup\gamma_2$ of distinct sizes, and the static vertical ray shooting data structure
 for each subset in $\bb(\gamma_1\cup\gamma_2)$.
 Then we insert $e$ to the data structure as we did in Case~2.

\begin{lemma}
	The total time for updating every vertical ray shooting data structure in the course of $n$ edge insertions is $O(n\log n\log\log n)$.
\end{lemma}
\begin{proof}
	Recall that the vertical ray shooting data structure $\ds(\cdot)$
	itself is static. 
	We charge the time for merging two subsets of size $2^i$ to
	the edges in the two subsets. Since this takes $O(2^{i+1}\log\log n)$ time, we charge each edge $\log\log n$ units.
	Then we show that each edge of $\Pi$ is charged at most $O(\log n)$ times in total in the course of $n$ edge insertions, which implies that the total update time is $O(n\log n\log\log n)$.
	
	Consider an edge $e$ in $\Pi$. When it is inserted to $\Pi$, it belongs to a subset of a connected component of $\Pi$ of size $1$. If it is charged once, the size of the subset which $e$ belongs to is doubled.
	Therefore, each edge of $\Pi$ is charged at most $O(\log n)$ times,
	and the total update time is $O(n\log n\log\log n)$. 
\end{proof}

\subsubsection{Rebalance Procedures for Trees}
Now we get rid of the assumption that the edges
to be inserted are known in advance. Thus we are required to
rebalance the trees we maintain for $\locatecc(\gamma)$.
We maintain the backbone tree and $\ds(\beta)$ for each subset $\beta$ of $\bb$.
Here, $\ds(\cdot)$ is a static structure obtained by the backbone tree, 
so it is balanced at any time. However, we are required to update it as the backbone tree is updated. This is because Lemma~\ref{lem:const-contracted-tree} requires
each non-dummy node of $\ds(\beta)$ to point to its corresponding node in the base tree.
In the following, we show how to rebalance the base tree of the backbone tree
and the segment trees associated with the nodes of the base tree.
Then we show how to update the pointer associated with each non-dummy node of the 
base tree of $\ds(\cdot)$.

\paragraph{Weight-balanced B-trees.}
We apply standard technique for keeping trees balanced using weight-balanced B-trees~\cite{weight-balance-Arge,weight-balance-Giora} with fan-out $f\geq 2$. It is a search tree with fan-out $f$  storing points at its leaves satisfying that all leaves are at the same distance from the root and $f^{h(v)}/2 \leq n(v)\leq f^{h(v)}$ for each node $v$,
where $h(v)$ is the height of $v$ and $n(v)$ is the size of the subtree
rooted at $v$. The height of a weight-balanced B-tree with fan-out  $f$ is $O(\log n/\log f)$.

When an element (point) is inserted to a weight-balanced
B-tree, we make a new leaf corresponding to the point and put it in an appropriate position. Then we consider each node $v$ in the path from 
the root to this leaf and check if it violates that 
$f^{h(v)}/2 \leq n(v)\leq f^{h(v)}$. If so, we split $v$ into two nodes $v_1$ and $v_2$. More specifically, we split the set of 
the children of $v$ into
two groups. We make the nodes in the first group be the children of $v_1$,
and the other nodes be the children of $v_2$. Then we make the parent of $v$
be the parent of $v_1$ and $v_2$. We also construct the secondary structure
associated with $v_1$ and $v_2$, and update the secondary structures associated
with the other nodes accordingly.
If we can apply this \emph{split} operation in $O(mA(m))$ time for each node,
where $m$ is the size of the subtree rooted in the node, 
the amortized time of inserting an element to the weight-balanced B-tree 
is $O(\log n+ A(n)\log n/\log\log n)$~\cite[Theorem 2.3]{weight-balance-Giora}.

\paragraph{Rebalancing the base tree.}

\begin{figure}
	\begin{center}
		\includegraphics[width=0.7\textwidth]{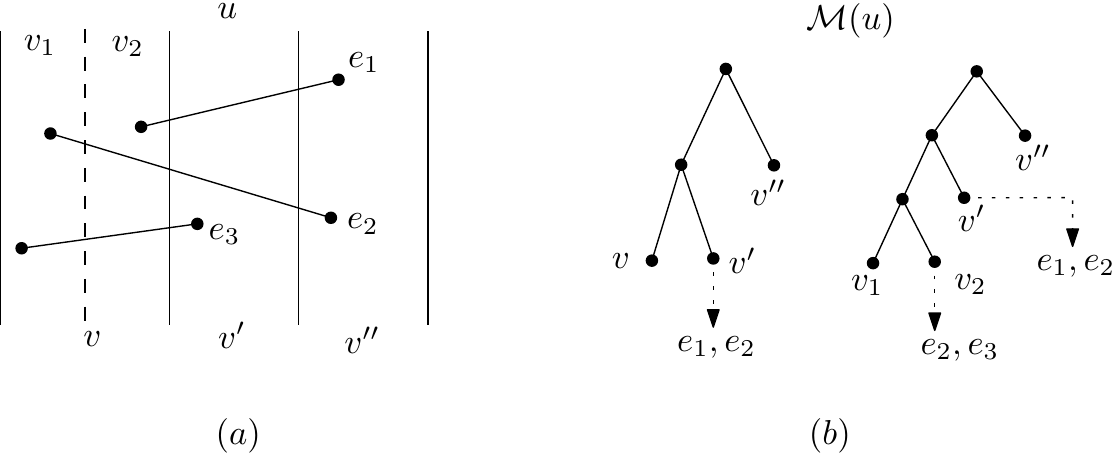}
		\caption{\small (a) Node $v$ is split into $v_1$ and $v_2$.
			Before the split of $v$, all three edges are in $\mathcal{L}(v)$.
			After the split of $v$, two of them, $e_2$ and $e_3$, are in
			$\mathcal{L}(v_1)$, and the other is in $\mathcal{L}(v_2)$.
			(b) Segment trees constructed on $\mathcal{M}(u)$, where $u$ is the parent of $v$. 
			Since $u$ gets two new children $v_1$ and $v_2$ instead of $v$,
			new two leaf nodes appear. The edges in $\mathcal{L}(v_1)$
			are associated with the new node corresponding to $v_2$.
			\label{fig:base-tree}}
	\end{center}
\end{figure}

In our problem, we maintain the base tree of the backbone tree
as a weight-balanced B-tree of fan-out $\log^\epsilon n$.
We show that we can split a node $v$ into two nodes $v_1$ and $v_2$ in 
 $O(m\log\log n)$ time, where $m$ is the size of the subtree rooted at the node.
The node $v$ has three lists of (pieces of) edges: $\mseg(v)$, $\lseg(v)$ and
$\rseg(v)$. All segments contained in $\mseg(v)$, $\lseg(v)$ or
$\rseg(v)$ have their endpoints in the leaf nodes of the subtree rooted at $v$. 
Thus $m$ is $O(|\mseg(v)|+|\lseg(v)|+|\rseg(v)|)$.
We compute the secondary structures in three steps.

First, we compute $\lseg(v_1)$ and $\lseg(v_2)$.
The structure on $\lseg(v)$ is a sorted list for each subset in $\bb$.
An edge that was in $\lseg(v)$ is contained in $\lseg(v_1)$ or $\lseg(v_2)$
after $v$ is split into $v_1$ and $v_2$. 
We split each sorted list into two sublists such that one sublist
consists of edges having their left endpoints in $\region{v_1}$ and the other sublist
consists of edges having their left endpoints in $\region{v_2}$. See Figure~\ref{fig:base-tree}(a). This takes 
time linear in the number of the edges of $\lseg(v)$.
This completes the computation of $\lseg(v_1)$ and $\lseg(v_2)$.
We also do this for $\rseg(v)$.

Second, we update $\mseg(u)$ for the parent $u$ of $v$.
If an edge $e$ that was in $\lseg(v)$ is assigned to $\lseg(v_1)$, the piece corresponding to $e$
stored in $\mseg(u)$ changes, where $u$ is the parent of $v$ (and thus the parent of $v_1$ and $v_2$.) 
See the edges $e_2$ and $e_3$ in Figure~\ref{fig:base-tree}(a).
More specifically, the piece of $e$ in $\region{v_2}$ is added to the piece of $e$ that were stored in $\mseg(u)$ before $v$ is split. 
 We update $\mseg(u)$ accordingly as follows. 
The structure on $\mseg(u)$ is a segment tree with associated 
sorted lists.  As $v$ is split into $v_1$ and $v_2$,
a leaf node in the segment tree is also split into two nodes.
See Figure~\ref{fig:base-tree}(b).
 But all edges, except the ones in $\lseg(v_1)$ (and in $\rseg(v_2)$), are 
still stored in the nodes where they are stored before $v$ is split. 
Moreover, by the construction of the segment tree,
the edges in $\lseg(v_1)$ are stored in the new leaf node corresponding 
to $v_2$, and in the nodes where they are stored before $v$ is split. 
 We also do this for $\rseg(v_2)$. 
Thus we can update $\mseg(u)$ in time linear in $O(|\lseg(v)|)$.

Third, we compute $\mseg(v_1)$ and $\mseg(v_2)$.
After the split of $v$, an edge of $\mseg(v)$ is in 
$\mseg(v_1)$ or $\mseg(v_2)$. 
See Figure~\ref{fig:base-tree-middle}(a). 
We construct
the segment trees on $\mseg(v_1)$ and $\mseg(v_2)$. 
They are the subtrees of the segment tree of $\mseg(v)$ rooted
at the children of the root node. Thus we just copy them for 
$\mseg(v_1)$ and $\mseg(v_2)$.
See Figure~\ref{fig:base-tree-middle}(b--c).
Thus it takes time linear in the
total size of them, which is $O(m\log\log n)$.

\begin{figure}
	\begin{center}
		\includegraphics[width=\textwidth]{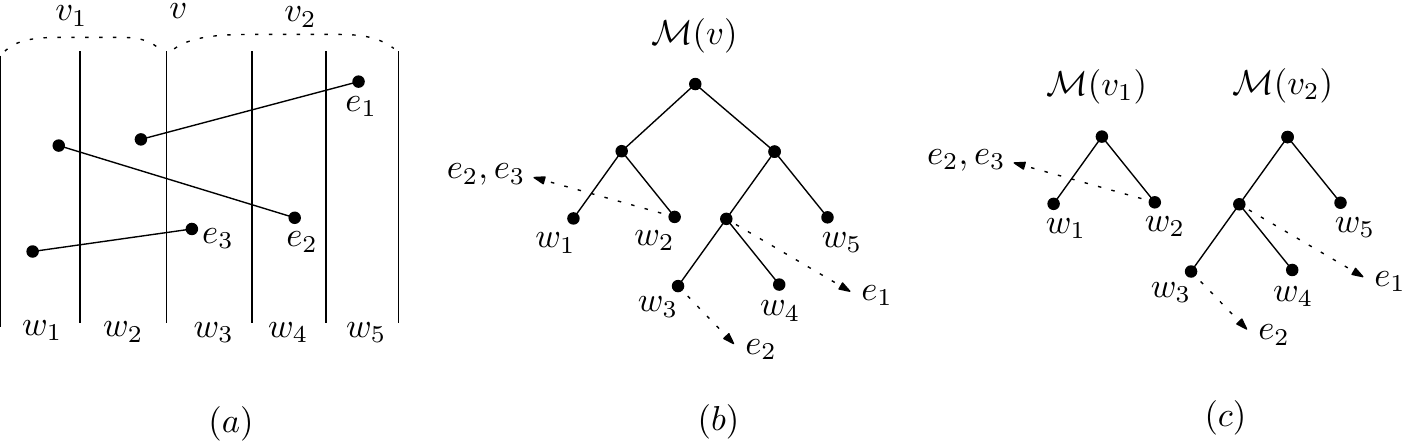}
		\caption{\small (a) After $v$ is split into $v_1$ and $v_2$,
			the edge $e_1$ belongs to $\mseg(v_2)$,
			the edge $e_2$ belongs to no $\mseg(\cdot)$,
			and the edge $e_3$ belongs to $\mseg(v_1)$.
			(b) To obtain the segment trees of $\mathcal{M}(v_1)$
			and $\mathcal{M}(v_2)$, it suffices to remove the root of the segment tree of $\mathcal{M}(v)$. 
			(c) The two resulting trees
			are the ones for $\mathcal{M}(v_1)$
			and $\mathcal{M}(v_2)$.
			\label{fig:base-tree-middle}}
	\end{center}
\end{figure}
\paragraph{Update the pointer associated with nodes of $\ds(\cdot)$.}
As a node $v$ in the base tree is split into $v_1$ and $v_2$,
we also update the static data structures $\ds(\beta)$ whose node
points to $v$. We
are required to do this because we need to maintain
the pointers for each non-dummy node of a contracted backbone tree
pointing to its corresponding node in the backbone tree. (Refer to
Lemma~\ref{lem:const-contracted-tree}.) 

We find the nodes in $\ds(\cdot)$ corresponding to $v$ in time linear in the
reported nodes. Let $\bb'$ be the set of subsets $\beta$ of $\bb$ such that
$\ds(\beta)$ has such a node. 
Let $\beta$ be a subset in $\bb'$. We split the node of $\bb'$ pointing to $v$  into two nodes which correspond to $v_1$ and $v_2$, respectively, as we did for the base tree.
The time for updating $\ds(\beta)$ is $O(m_\beta\log\log n)$, where 
$m_\beta$ denotes the number of the edges of $\beta$ stored in $\lseg(v)$, $\mseg(v)$ and $\rseg(v)$. Notice that the sum of $m_\beta$ for every 
subset $\beta$ of $\bb'$ is $O(m)$. 
Therefore, the total time of the split operation on the node pointing to 
$v$ in $\ds(\beta)$ for every $\beta\in\bb'$ 
is $O(m\log\log n)$ time, where $m$ is the size of the subtree rooted at the node.

\paragraph{Rebalancing the segment trees.}
We maintain a binary segment tree of the edges of $\mseg(v)$ for each node $v$ of the base tree as a weight-balanced B-tree of constant-bounded fan-out.
Contrast to the base tree (interval tree), the size of the subtree rooted at a node $u$ is not bounded
by the number of the edges stored in $u$ in this case. Here, we will see that 
a split operation
can be done in constant time, which leads to the total update time of $O(\log n\log\log n)$.

A segment tree on a set $\mathcal{I}$ of intervals on the $x$-axis 
of fan-out $f$ is defined as follows. Its base tree is a search tree of fan-out $f$  on the 
endpoints of the intervals of $\mathcal{I}$. An interval $e$ of $\mathcal{I}$ is stored in $O(\log n/\log f)$ nodes $u$ such that 
$\region{w}\subset e$ and $\region{u}\not\subset e$ for the parent $w$ of $u$.
 Notice that $e$ is stored in at most one node of the same depth. 
Let $w_1,\ldots, w_{f'}$ be the children of a node $u$ such that $\region{w_i}$ lies to the left of $\region{w_j}$ for $i\leq j$.
Each node $u$ has $f^2$ lists $L_{ij}$ of edges stored in $u$ for $1\leq i\leq f'$ and $i\leq j\leq f'$.
The list $L_{ij}$ consists of the edges stored in $u$ crossing
 $\region{w_i}\cup\cdots\cup \region{w_j}$. 
The total space complexity of the segment tree is $O(f^2n)$.

We show how to split a node $u$ into two nodes $u_1$ and $u_2$.
Let $w_1,\ldots, w_{f'}$ be the children of $u$. Let $f''=\lfloor f'/2\rfloor$. 
We make $w_i$ be a child of $u_1$ if $1\leq i\leq f''$, 
and make it a child of $u_2$, otherwise. 
Among ${f'}\choose{2}$ lists associated with $u$, the lists $L_{ij}$ are assigned to $u_1$
if $1\leq i\leq j\leq f''$,
and assigned to $u_2$ if $f''\leq i\leq j\leq f'$.
This can be done in $O(f^2)$ time in total, which is a constant time in our case.
Therefore, the split operation takes a constant time, and we have the following lemma.

\begin{lemma}
	We can maintain a data structure of size $O(n\log\log n)$
	in an incremental planar subdivision $\Pi$
	so that the edge of $\gamma$ 
	lying immediately above $q$ can be found in $O(\log^2 n)$ time
	for any edge $e$ and any connected component $\gamma$ of $\lmm$.
	The update time of this data structure is $O(\log n\log\log n)$.
\end{lemma}

\section{Incremental Stabbing-Lowest Data Structure for Trapezoids}\label{sec:stabbing}
In this section, we are given a set $\mathcal{T}$ of trapezoids which is initially empty. Then we are to process the insertions of trapezoids to $\mathcal{T}$ so that
the lowest trapezoid in $\mathcal{T}$ stabbed by a query point can be found
efficiently.
We present a data structure of size $O(n\log n)$
supporting $O(\log^2n)$ query time
and $O(\log n\log\log n)$ insertion time.

\subsection{Data Structure}\label{sec:large-size}
The data structure we present in this subsection is an interval tree of fan-out $f=\log^\epsilon n$ 
of the upper and lower sides of the trapezoids of $\mathcal{T}$, where $\epsilon$ is an arbitrary fixed constant 
with $0<\epsilon <1$. 
Since the left and right sides of the trapezoids are parallel to the $y$-axis,
a node of the interval tree stores the upper side of a
trapezoid of $\mathcal{T}$ if and only if it stores the lower side
of the trapezoid.
Here, instead of storing the upper and lower sides of a trapezoid,
we store the trapezoid itself in such a node. 
In this way, a trapezoid $\Box$ of $\mathcal{T}$ is stored in at most three nodes of the interval tree: $\lseg(u), \rseg(u')$ and $\mseg(v)$ for two children $u$ and $u'$ of a node $v$. 
For details, refer to Section~\ref{sec:preliminary}.

\paragraph{Secondary structure for $\lseg(v)$ (and $\rseg(v)$).}
We first describe the secondary structure only for $\lseg(v)$ for a node $v$ of the base tree. The structure for $\rseg(v)$ can be defined and constructed analogously.
By construction, every trapezoid of $\lseg(v)$ intersects the 
right vertical line $\ell(v)$ on the boundary of $\region{v}$.
Thus, their upper and lower sides can be sorted in their $y$-order.
See Figure~\ref{fig:trapezoid-query}. 
Let $\mathcal{I}(v)$ be the set of the intersections of the trapezoids
of $\lseg(v)$ with $\ell(v)$. Note that it is a set of intervals of $\ell(v)$. 
We construct a binary segment tree $T(v)$ of $\mathcal{I}(v)$. 
A node $u$ of $T(v)$ corresponds to an interval $\region{u}$ contained in $\ell(v)$.
Every interval of $\mathcal{I}(v)$ stored in $u$
contains $\region{u}$. An interval $I\in \mathcal{I}(v)$ has its corresponding trapezoid $\Box$ in $\lseg(v)$ such that $\Box\cap \ell(v)= I$.
We let $I$ have the \emph{key} 
which is the $x$-coordinate of the left side of $\Box$.

\begin{figure}
	\begin{center}
		\includegraphics[width=0.5\textwidth]{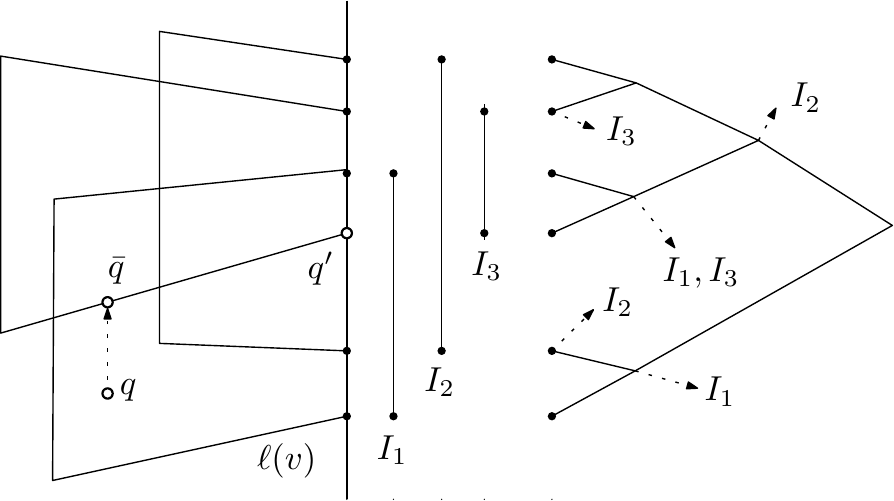}
		\caption{\small The segment tree constructed on
			the intersections of the trapezoids of $\mathcal{L}(v)$
			with $\ell(v)$.
			\label{fig:trapezoid-query}}
	\end{center}
\end{figure}

We construct an associated data structure
so that given a query value $x$ the interval 
with lowest upper endpoint can be found efficiently 
among the intervals stored in $u$ and having their keys less than $x$.
Imagine that we sort the intervals of $\mathcal{I}(v)$ stored in $u$ with respect to their keys,
and denote them by $\langle I_1, \ldots, I_k\rangle$. 
And we use $\Box_i\in\mathcal{T}$ to denote the trapezoid corresponding
to the interval $I_i$ (i.e., $\ell(v)\cap \Box_i = I_i$) for $i=1,\ldots,k$. 
The associated data structure is just a sublist of $\langle I_1,\ldots, I_k\rangle$. 
Specifically, suppose $x$ is at least the key of $I_i$ and at most
the key of $I_{i+1}$ for some $i$.
Then every interval in $\langle I_1,I_2,\ldots, I_{i}\rangle$ 
has its key at most $x$.
Thus the  answer to the query is the one with lowest upper endpoint among  $\langle I_1,I_2,\ldots, I_{i}\rangle$.
Using this observation, we construct a sublist of $\langle I_1,\ldots, I_{k}\rangle$ as follows. 
We choose the interval, say $I_i$, if its upper endpoint is the lowest among the upper endpoints of the intervals 
in $\langle I_1,\ldots I_i\rangle$.
We maintain the sublist consisting of the chosen intervals. 
Notice that the sublist has \emph{monotonicity} with respect to their upper endpoints. That is, the upper endpoint of $I_{i}$ lies lower than the
upper endpoint of $I_{i'}$ if $I_i$ comes before $I_{i'}$ in the sublist.
This property makes the update procedure efficient.

By applying binary search on the sublist with respect to the keys,
 we can find the interval
with lowest endpoint 
among the intervals stored in $u$ and having the keys less than $x$.
For each node of the base tree,
we maintain a structure for dynamic fractional cascading~\cite{Mehlhorn1990} on the segment tree so that 
the binary search on the sublist associated with each node of the segment tree can be done in $O(\log n\log\log n)$ time
in total.

A tricky problem here is that 
a query point $q$ and the upper or lower side of a trapezoid in $\lseg(v)$
cannot be ordered with respect to the $y$-axis in general. This happens if 
the left side of the trapezoid lies to the right of $q$. See Figure~\ref{fig:trapezoid-query}. This makes it difficult to follow
the path from the root to a leaf node in the segment tree associated with $\lseg(v)$. To resolve this, we will find the side $e$ lying immediately
above $q$ among the upper and lower sides of the
trapezoids in $\lseg(v)$, and then follow the path from the root
to the leaf corresponding to $q'=e\cap\ell(v)$.
In Section~\ref{sec:stabbing-query}, we will see why this gives the correct answer. 
To do this, we construct a vertical ray shooting data structure
on the upper and lower sides of the trapezoids in $\lseg(v)$. 
Since all of them intersect a common vertical line, we can use
a priority search tree as a vertical ray shooting data structure.

\paragraph{Secondary structure for $\mseg(v)$.}
Now we construct a secondary structure for $\mseg(v)$.
It is a binary segment tree on
the projections of the trapezoids of $\mseg(v)$ onto the $x$-axis.
Note that the height of the binary segment tree is $O(\log\log n)$. 
The trapezoids stored in each node $u$ of the segment tree 
intersect the vertical slab $\region{u}$, and have their left and right
sides lying outside of $\region{u}$.
Therefore, their upper and lower sides can be sored in $y$-order.
We consider the intersections (intervals) of the trapezoids stored in $u$
with the left side of $\region{u}$. 
We construct a data structure for answering stabbing-min queries
on such intervals whose keys are their upper endpoints. 
It has size of $O(N)$, and supports
$O(\log N)$ query time and $O(\log N)$ update time, where $N$ is the number of the intervals stored in the data structure~\cite{Agarwal-stabbing-max}.

Here we also have the issue similar to the one we encountered for $\lseg(v)$: a query point $q$ and an interval stored in a node $u$ of the segment tree 
cannot be ordered with respect to the $y$-axis.
 This makes it difficult to apply a stabbing-min query on the intervals.
 To resolve this, we will find the side $e$ lying immediately
above $q$ among the upper and lower sides of the
trapezoids stored in the node, and then apply a query with $q'=e\cap\ell(v)$.
To do this, we construct a balanced binary search tree
on the upper and lower sides of the trapezoids stored in $u$. 

\paragraph{Space complexity.}
Now we analyze the space complexity of the data structure.
The size of the base tree is $O(n)$. Each trapezoid of $\mathcal{T}$ is stored in
at most three nodes of the base tree. 

Consider the structures on $\lseg(v)$ for a node $v$ of the base tree.
Each trapezoid of $\lseg(v)$ is stored in $O(\log n)$ nodes in 
the secondary structure (segment tree) constructed on the intersections of the trapezoids of $\lseg(v)$ with $\ell(v)$.
The size of the list associated with each node $u$ of the segment tree is linear
in the number of the trapezoids stored in $u$. Thus the total size of the 
data structures associated with a node $v$ 
is $O(|\lseg(v)|\log |\lseg(v)|)$. Also, the priority search tree
has a size linear in the segments stored in the data structure.
Thus the size of the secondary and tertiary  structures
on $\lseg(\cdot)$ for all nodes of the base tree is $O(n\log n)$.
Similarly, we can show that the total size of the secondary and tertiary structures on $\rseg(\cdot)$ is 
$O(n\log n)$.

Now consider the structures on $\mseg(v)$ for a node $v$ of the base tree.
A trapezoid of $\mseg(v)$ is stored in $O(\log\log n)$ nodes in the segment tree
associated with $v$. Moreover, the tertiary structure on each node $u$ of the segment tree 
is linear in the number of the trapezoids stored in $u$. Thus the size of the secondary and tertiary  structures
on $\mseg(\cdot)$ for all nodes of the base tree is $O(n\log\log n)$,
and therefore the total size of the data structure is $O(n\log n)$.

\subsection{Query Algorithm}\label{sec:stabbing-query}
Using this data structure, we can find the lowest trapezoid in $\mathcal{T}$ stabbed by a query point $q$ as follows. 
We follow the base tree (interval tree of fan-out $\log^\epsilon n$) along the 
search path $\pi$ of $q$. There are $O(\log n/\log\log n)$ nodes in $\pi$. 
For each node $v$, we consider the secondary structures 
on $\lseg(v)$, $\mseg(v)$ and $\rseg(v)$, and 
we find the lowest trapezoid stabbed by $q$ among the trapezoids
in each of them.
And we return the lowest one among all trapezoids we obtained from the nodes of $\pi$. We spend $O(\log n\log\log n)$ time on each node in $\pi$,
which leads to the total query time of $O(\log^2 n)$.

\paragraph{Searching on $\lseg(v)$ (and $\rseg(v)$) for a node $v$.}
We have a binary segment tree on the intersections of the trapezoids of $\lseg(v)$
with $\ell(v)$ for a node $v$ in $\pi$. Also, we have a priority search tree for the upper and lower sides of the trapezoids of $\lseg(v)$.  
We first find the side $e$ immediately lying above $q$ among them 
in $O(\log n)$ time
using the priority search tree, and let $q'$ be the intersection
point between $e$ and $\ell(v)$. See Figure~\ref{fig:trapezoid-query}.
The following lemma is a key for the query algorithm for $\lseg(v)$.

\begin{lemma}
	The lowest trapezoid stabbed by $q$ is stored in a node
	in the search path of $q'$. 
\end{lemma}
\begin{proof}
	Let $\bar{q}$ be the point on the upper or lower side of a trapezoid
	 where the vertical upward ray from $q$
	first hits, and $e$ be the 
	side containing both $\bar{q}$ and $q'$.
	By construction, every trapezoid containing $q'$
	is stored in a node in the search path $\pi_v$ of $q'$.
	
	We claim a trapezoid $\Box$ of $\lseg(v)$ containing $q$ also contains 
	$q'$. If this claim holds, every trapezoid of $\lseg(v)$ containing $q$
	is stored in a node of $\pi_v$, and thus the lemma holds.
	By construction, the line segment $q\bar{q}$ is not crossed by any
	side of a trapezoid of $\lseg(v)$. Thus $\Box$ contains $q\bar{q}$.
	Note that the segment $\bar{q} q'$ is contained in the lower or upper
	side of a trapezoid of $\lseg(v)$. Therefore, it is not crossed  by any
	upper or lower side of a trapezoid of $\lseg(v)$.  
	It is also contained in $\Box$.
	Therefore, $q'$ is contained in $\Box$.
	Therefore, the lemma holds.	
\end{proof}

By the lemma, it suffices to 
consider $O(\log n)$ nodes $w$ in the segment tree with $q'\in \region{w}$.
Then we find the successor of the $x$-coordinate of $q$ 
on the sublist associated with each such node.
By construction, the trapezoid corresponding to the successor is the lowest trapezoid
stabbed by $q$ among all trapezoids stored in $w$. 
Using (dynamic) fractional cascading, we can find it in $O(\log\log n)$ time 
for each node after spending $O(\log n)$ time for the initial binary search of only one node in the segment tree.
Thus we can find all successors in $O(\log n\log\log n)$ time.
We also do this for $\rseg(v)$ in $O(\log n\log\log n)$ time.

 \paragraph{Searching on $\mseg(v)$ for a node $v$.}
As the structure on $\mseg(v)$ for a node $v$ in $\pi$, we have a segment tree
constructed on $\mseg(v)$ with respect to the $x$-coordinates of their left and right
sides. Thus its height is $O(\log\log n)$.
We find $O(\log\log n)$ nodes $w$ of the segment tree with $q\in \region{w}$.
For each such node, we find the side $e$ lying immediately above $q$
in $O(\log n)$ time using the associated binary search tree.
Let $q'$ be the intersection point of $e$ with $\ell(v)$. Then
we apply a stabbing-min query with $q$ in $O(\log n)$ time.
By construction, the trapezoid corresponding the interval stabbed by $q'$ with lowest
key value is the lowest trapezoid stabbed by $q$.
This takes $O(\log n\log\log n)$ time in total for the node $v$.
Therefore, the total query time is $O(\log^2 n)$.

\begin{lemma}
	Using the data structure described in Section~\ref{sec:large-size},
	we can find the lowest trapezoid stabbed by a query point in $O(\log^2 n)$ time.
\end{lemma}

\subsection{Update Algorithm}\label{sec:stabbing-update}
We assume that the trapezoids
to be inserted are known in advance so that we can keep the base tree 
and all segment trees balanced.
We can get rid of this assumption using weight-balanced B-trees as we did
in Section~\ref{sec:D2}.
Let $\Box$ be a trapezoid to be inserted to the data structure.
We find the node $v$ of maximum depth
in the base tree such that $\region{v}$ contains $\Box$
in $O(\log n)$ time. Then we find two children $u$ and $u'$ of $v$ such that
$\region{u}$ contains the left side of $\Box$ and $\region{u'}$ contains the right
side of $\Box$.
The trapezoid $\Box$ is to be stored only in these nodes in the base tree.

\paragraph{Update the segment tree on $\lseg(u)$ (and $\rseg(u')$).}
We update the secondary structure (segment tree) for $\lseg(u)$ by inserting $\Box$.
We find the set $W$ of $O(\log n)$ nodes $w$ in the segment tree such that the union of $\region{w}$'s
contains $\Box\cap \ell(u)$ and $\region{w}$'s are pairwise interior disjoint. Recall that the segment tree is constructed in 
the interval $\Box'\cap \ell(u)$
for every trapezoid $\Box'$ in $\lseg(u)$.
Each such node $w$ is associated with a sorted list $L(w)$ of intervals stored in $w$. We decide if we store $\Box\cap \ell(u)$ in $L(w)$. To do this, we find the position for $\Box$ in $L(w)$
by applying binary search on $L(w)$ with
respect to the key (the $x$-coordinate of the left side of $\Box$) of $\Box$. Here we do this for every node in $W$, and thus we can apply
binary searches more efficiently using fractional cascading.
The key of each interval in the sorted lists is a real number, which
is the $x$-coordinate of the left side of a trapezoid.
Thus we can apply (dynamic) fractional cascading so  that
each binary search takes $O(\log\log n)$ time after spending $O(\log n)$
time on the initial binary search on a node of $W$~\cite{Mehlhorn1990}.

Let $\langle I_1,\ldots, I_k\rangle$ be the sorted list of the intervals stored in $w$. The list $L(w)$ is a sublist of this list, say
$\langle I_{i_1},\ldots,I_{i_t}\rangle$. Let $I_{i_j}$ be the 
predecessor of $\Box\cap \ell(u)$. 
We determine if $\Box$ is inserted to the list 
in constant time: if the upper side of $\Box$ lies below the upper side of the trapezoid $\Box_{i_{j+1}}$ with $\Box_{i_{j+1}}=I_{i_{j+1}}\cap \ell(u)$,
we insert $\Box\cap\ell(u)$ to the list. 
Otherwise, the list stored in $u$ remains the same.
If we insert $\Box\cap\ell(u)$ to the list,
we check if it violates the monotonicity of $L(w)$.
We consider the trapezoid $\Box'$ whose corresponding interval lies 
before $\Box$ one by one from $\Box_{i_j}$. If 
the upper side of $\Box'$ lies above $\Box$, we remove $\Box'$ from the list.
Each insertion into and deletion from $L(w)$ takes $O(\log\log n)$ time~\cite{Mehlhorn1990}. 
We do this until the upper side of $\Box'$ lies below the upper side of  $\Box$.
The total update time for the insertion of $\Box$ is $O(\log n+ N\log\log n)$,
where $N$ is the number of the trapezoids deleted due to $\Box$.
Notice that a trapezoid (more precisely, its corresponding interval)  appears in $O(\log n)$ lists associated with nodes in the
segment trees. Also, each of them is deleted by at most once. 
Therefore, the sum of $N$ over all $n$ insertions is $O(n\log n)$.
Thus the amortized time for inserting $\Box$ to $\lseg(u)$ is $O(\log n\log\log n)$ time.
We do this for $\lseg(u')$ in $O(\log n\log\log n)$ time analogously.

\paragraph{Update the segment tree on $\mseg(v)$.}
Now we insert $\Box$ to the data structure for $\mseg(v)$. By the assumption we made at the beginning of this subsection, 
the segment tree on $\mseg(v)$ is balanced. 
Thus it suffices to find $O(\log\log n)$ nodes $w$ in the segment tree 
such that the union of $\region{w}$'s
contains $\Box$ and $\region{w}$'s are pairwise interior disjoint. 
For each such node $w$, we insert the intersection of $\Box$ with the left side  $\region{w}$ to the stabbing-min data structure associated with $w$,
which takes $O(\log n)$ time.
Therefore, the update time for $\mseg(v)$ is $O(\log n\log\log n)$,
and the total update time is also $O(\log n\log\log n)$.

\begin{lemma}
	We can maintain an $O(n\log n)$-size data structure on an incremental set of $n$ trapezoids supporting $O(\log n\log\log n)$ amortized update time 
	so that given a query point $q$, the lowest trapezoid stabbed by $q$ can be computed in $O(\log^2 n)$ time.
\end{lemma}

\bibliography{paper}
 
\end{document}